\documentclass[conference]{IEEEtran}
\IEEEoverridecommandlockouts

\usepackage{comment}
\usepackage{caption}
\usepackage{varwidth}
\usepackage{amsmath}
\usepackage{textcomp}
\usepackage{hyperref}
\hypersetup{colorlinks=true,citecolor = blue}

\newtheorem{definition}{Definition}

\newtheorem{theorem}{Theorem}

\newtheorem{upper bound}{Upper bound}
\newtheorem{strategy}{Strategy}    
  
\usepackage[noend]{algpseudocode}
\usepackage{graphicx,epstopdf,algpseudocode,caption,url}   
\graphicspath{{./images}}

\usepackage{amssymb} 
\usepackage{array}    
\usepackage{arydshln} 
\usepackage{graphics}
\usepackage{color}   
\usepackage{graphicx}  
\usepackage{graphicx,epstopdf,caption,url}   
\usepackage{multirow} 
\newenvironment{proof}{\begin{IEEEproof}}{\end{IEEEproof}} 
\usepackage[ruled,linesnumbered]{algorithm2e}
\usepackage[switch]{lineno}  	
	
\begin{document}

\title{Discovering High-utility Sequential Rules with Increasing Utility Ratio}

\author{Zhenqiang Ye, Wensheng Gan*, Gengsen Huang, Tianlong Gu, Philip S. Yu,~\IEEEmembership{Life Fellow,~IEEE} 

\thanks{This research was supported in part by the National Natural Science Foundation of China (No. 62272196), Guangzhou Basic and Applied Basic Research Foundation (No. 2024A04J9971), Engineering Research Center of Trustworthy AI, Ministry of Education (Jinan University), and Guangdong Key Laboratory of Data Security and Privacy Preserving. Corresponding author: Wensheng Gan}

\thanks{Zhenqiang Ye, Wensheng Gan, Gengsen Huang, and Tianlong Gu are with the College of Cyber Security, Jinan University, Guangzhou 510632, China. (E-mail: yzq66f@gmail.com, wsgan001@gmail.com, hgengsen@gmail.com, gutianlong@jnu.edu.cn)}
	
\thanks{Philip S. Yu is with the University of Illinois Chicago, Chicago, USA. (E-mail: psyu@uic.edu)} 
}	
		
\maketitle

\begin{abstract}
  Utility-driven mining is an essential task in data science, as it can provide deeper insight into the real world. High-utility sequential rule mining (HUSRM) aims at discovering sequential rules with high utility and high confidence. It can certainly provide reliable information for decision-making because it uses confidence as an evaluation metric, as well as some algorithms like HUSRM and US-Rule. However, in current rule-growth mining methods, the linkage between HUSRs and their generation remains ambiguous. Specifically, it is unclear whether the addition of new items affects the utility or confidence of the former rule, leading to an increase or decrease in their values. Therefore, in this paper, we formulate the problem of mining HUSRs with an increasing utility ratio. To address this, we introduce a novel algorithm called SRIU for discovering all HUSRs with an increasing utility ratio using two distinct expansion methods, including \textit{left-right} expansion and \textit{right-left} expansion. SRIU also utilizes the item pair estimated utility pruning strategy (\textit{IPEUP}) to reduce the search space. Moreover, for the two expansion methods, two sets of upper bounds and corresponding pruning strategies are introduced. To enhance the efficiency of SRIU, several optimizations are incorporated. These include utilizing the Bitmap to reduce memory consumption and designing a compact utility table for the mining procedure. Finally, extensive experimental results from both real-world and synthetic datasets demonstrate the effectiveness of the proposed method. Moreover, to better assess the quality of the generated sequential rules, metrics such as confidence and conviction are employed, which further demonstrate that SRIU can improve the relevance of mining results. By enforcing a strict requirement that each rule extension must increase (or maintain) its utility ratio relative to its parent, SRIU uniquely identifies value-accretive sequential patterns. This capability makes it particularly suitable for decision-critical domains where actionable insights must guarantee progressive value gain—such as personalized e-commerce recommendations, financial risk monitoring, and clinical pathway analysis. The source code is publicly available on GitHub: https://github.com/DSI-Lab1/SRIU.
\end{abstract}
		
		
\IEEEpeerreviewmaketitle

\section{Introduction} \label{sec:introduction}

It is well known that data is growing at an exponential rate due to the widespread adoption of digital technologies and the internet \cite{sun2023internet,gan2023web}. A wide variety of applications can generate data that contains important information. However, despite the vast amount of data generated and collected, only a small portion is readily available for analysis. Therefore, employing data mining techniques to extract valuable information and knowledge from the data is imperative. Pattern mining is one of the data mining techniques. In this field, frequency serves as a crucial metric for evaluating mining results. Frequent pattern mining (FPM) \cite{luna2019frequent,fournier2022pattern} is used to discover frequent patterns in a database. That is, if a pattern is frequent, the support of this pattern reaches the minimum support (\textit{minsup}) threshold. FPM has been used successfully in many fields, such as market basket analysis, bioinformatics \cite{birzele2006new}, and social network analysis \cite{moosavi2017community}. However, one of its limitations lies in its inability to reflect the temporal order of events occurring in the real world. A prime example of this limitation is the frequent pattern \{monitor, desktop computer\} observed during the school season on an e-commerce platform. In reality, typical consumer behavior involves purchasing a desktop first and then considering the monitor. If both of them appear in a transaction, it means the pattern in the real world is more like \{desktop computer, monitor\}. To address this issue, sequential pattern mining (SPM) \cite{agrawal1995mining,gan2019survey} has been proposed. SPM can be used to handle the sequence data, which considers the chronological order of items in the database's sequences. As reviewed in \cite{fournier2017survey}, the applications of SPM are more suitable for practical scenarios such as text analysis, webpage click-stream analysis, e-learning, and others.

Although the results of SPM can provide deep insights into the data and aid in strategy development, users are still uncertain about the reliability of patterns mined by the SPM algorithm. For instance, consider the frequent sequential pattern $\alpha$ = \{desktop computer, monitor, brake pads\}. While this pattern may be frequent, recommending brake pads after users have purchased desktops and monitors might seem particularly unusual and inappropriate. This observation suggests that SPM could potentially produce some misleading results in certain contexts. To address this issue, researchers introduced the support-confidence framework to SPM, which is called sequential rule mining (SRM) \cite{fournier2011rulegrowth, fournier2012cmrules,gan2025towards}. SRM aims to find all sequential rules (SRs) that meet predefined thresholds, and the rule format can be presented as $\{X\} \to \{Y\}$. For example, an SR \{computer\} $\to$ \{keyboard\} is a rule whose support and confidence value meet the user-defined \textit{minsup} threshold and the minimum confidence (\textit{minconf}) threshold. Since SRM can more clearly describe the correlation between items, it is widely applied in various fields, such as predicting user adoption and weather forecasts. However, because the frequency indicator is limited, SRM also misses some useful information. In many real-world scenarios, items are characterized by diverse attributes or weights, such as profit margins in business or risk indices in emergency response systems. For example, an SR \{milk, bread\} $\to$ \{champagne\} can provide valuable insights into high-revenue-generating item combinations due to their high transaction volumes. However, some SRs like \{diamond\} $\to$ \{rose\} can also bring a high profit, and SRM will miss them since their support does not satisfy the \textit{minsup} threshold. To obtain a more comprehensive understanding of the relationships between items, it is imperative to consider these attributes. Therefore, the utility concept was introduced to SRM, and the task of high-utility sequential rule mining (HUSRM) \cite{gan2021survey,zida2015efficient,huang2023us} was proposed. HUSRM can discover all HUSRs with high utility values. Nevertheless, HUSRM also has some drawbacks. Given two HUSRs, $r_1$ = \{gamepad\} $\to$ \{computer game\} and $r_2$ = \{gamepad\} $\to$ \{computer game, video\}. $r_2$ can be perceived as an extension of $r_1$, achieved by the inclusion of the item \{video\}. What role does \textit{video} play in this, positive or negative? We cannot know by using the existing HUSRM algorithms.

To address this problem, in this paper, we propose a new algorithm called SRIU that discovers HUSRs under the increasing utility ratio setting. Inspired by some related studies \cite{zida2015efficient,huang2023us,zhang2024totally} in HUSRM, we adopt some efficient pruning strategies and design corresponding data structures. The main contributions of our work are summarized as follows:
		
\begin{itemize}
  \item To the best of our knowledge, there has been no prior research on the increasing relationship between the mined sequential rules. We propose an algorithm called discovering high-utility \underline{S}equential \underline{R}ules with \underline{I}ncreasing \underline{U}tility ratio (SRIU), which mines all satisfied HUSRs in a sequential database.
  
   \item We formulate the problem of mining HUSRs with an increasing utility ratio and then introduce the $e$-index-based method to determine the appropriate expansion strategies during the mining process. It is noteworthy that each expansion strategy entails distinct depth-pruning methods.
 
  \item We conduct extensive experiments on different real and synthetic datasets. Not only is the performance of SRIU tested, but also different expansion strategies are analyzed. The results show the superior performance of SRIU. Furthermore, to illustrate the effectiveness of SRIU, the mining results are evaluated under the indicators of confidence and conviction.
\end{itemize}	
		
The remainder of this paper is organized as follows: Section \ref{sec:relatedWork} briefly reviews the related work about HUSRs. Section \ref{sec:preliminaries} introduces some basic definitions and a problem statement. Section \ref{sec:algorithm} describes the proposed algorithm. Finally, experimental evaluations and analysis are shown in Section \ref{sec:experimental}, and the conclusion is presented in Section \ref{sec:conclusion}.

\section{Related Work}  \label{sec:relatedWork}
\subsection{Sequential Pattern Mining}

Sequential pattern mining (SPM) \cite{gan2019survey,fournier2017survey} was proposed to discover all patterns that satisfy a certain frequency constraint in the sequence database. Compared with frequent itemset mining (FIM), the results generated by SPM can carry more semantic information and provide higher practicability \cite{fournier2022pattern}. There are many representative algorithms. The first algorithm, GSP \cite{srikant1996mining}, employs a strategy of pre-pruning infrequent sequences by utilizing defined adjacency subsequences. However, GSP needs to scan the sequence database multiple times, thus consuming a lot of time and memory. To address this issue, the PrefixSpan algorithm \cite{han2001prefixspan} adopts a projected technique to record growing patterns. Since PrefixSpan needs to construct the projection database for each frequent prefix pattern, it consumes a lot of memory when the size of the sequence database is large or the number of items contained in the sequence database is considerable. The SPADE algorithm is based on the idea of equivalence classes \cite{zaki2001spade}, and the SPAM algorithm \cite{ayres2002sequential} is based on the bitmap structure. There are also many interesting studies of SPM with different constraints and applications, such as e-RNSP \cite{dong2020rnsp}, sc-NSP \cite{gao2023toward}, NTP-Miner \cite{wu2022ntp}, and SCP-Miner \cite{wu2022top}.

\subsection{High-Utility Sequential Pattern Mining}

To further study the semantic information contained in sequential patterns, the concept of utility in SPM tasks was introduced, and high-utility sequential pattern mining (HUSPM) was proposed to discover more realistic patterns. Different from SPM, HUSPM \cite{gan2021survey} cannot utilize the Apriori property \cite{agrawal1995mining} to prune unpromising patterns, which makes HUSPM more difficult than SPM. To reduce the search space, two primary ways can be taken. The first one is to design upper bounds; the other one is to devise suitable structures. Sequence weighted utility (\textit{SWU}) \cite{ahmed2010novel} is the first upper bound with a downward closure feature for HUSPM tasks, which can estimate the actual utility of patterns, thus pruning unpromising patterns in advance. That is, if the \textit{SWU} of a pattern is less than \textit{minUtil} (the minimum utility threshold), then the pattern is impossible to hold high utility. Based on the PrefixSpan, UtilitySpan \cite{ahmed2010novel} uses \textit{SWU} to discover all candidate patterns and then obtains the real HUSPs. The USpan algorithm \cite{yin2012uspan} uses the upper bounds \textit{SWU} and sequence projected utilization (\textit{SPU}) to reduce the search space. However, \textit{SPU} misses some HUSPs, and the lexicographic q-sequence tree (LOS-tree) mentioned in USpan also consumes massive memory. To further improve the efficiency of the algorithm, Wang \textit{et al.} \cite{wang2016efficiently} proposed the HUS-Span algorithm, which adopts two tighter upper bounds, the prefix extension utility (\textit{PEU}) and reduced sequence utility (\textit{RSU}). HUS-Span also utilizes the lexicographic tree to represent the space of HUSPM, but HUS-Span adopts a novel data structure called the utility chain to store the utility information that is stored in the utility matrix in USPan. Gan \textit{et al.} \cite{gan2020proum} proposed the ProUM algorithm, which improves the efficiency of mining HUSPs by combining the projected technique and the novel data structure called utility-array. They further proposed a more efficient algorithm called HUSP-ULL \cite{gan2020fast}, which adopted the utility-linked list (UL-list) to quickly generate the projected sequence databases. Moreover, two pruning strategies, LAR and IIP, are designed to reduce the search space, and HUSP-ULL has better performance in experiments. Other interesting topics about HUSPM have also been studied \cite{gan2021explainable}.

\subsection{Sequential Rule Mining}

Although HUSPM can help users understand the knowledge behind the data, it lacks the ability to reason. To solve this limitation, Zaki \textit{et al.} \cite{zaki2001spade} first introduced the support-confident frame to discover the sequential rules (SRs). A sequential rule consists of an antecedent and a consequent, represented in the form of $\{X\} \to \{Y\}$ and $X \cap Y$ = $\emptyset$. The sequential rule mining (SRM) can be divided into two fields: totally-ordered SRM \cite{lo2009non,pham2014efficient} and partially-ordered SRM \cite{fournier2011rulegrowth, fournier2015mining}. In the totally-ordered SRM, the items of the antecedent and the consequent in sequences must follow the order. As for the partially-ordered SRM, the items in the antecedent or consequent can be unordered. Normally, totally-ordered SRM has stricter constraints, and its application scenarios are more targeted. The partially-ordered SRM tends to interpret probabilistic relationships between two itemsets. Therefore, the partially-ordered SRM is more suitable for prediction or recommendation. There are some algorithms for mining partially-ordered SRs. For example, Fournier-Viger \textit{et al.} proposed CMRules \cite{fournier2012cmrules} and CMDeo \cite{fournier2012cmrules}. Both CMRules and CMDeo are two-phase-based algorithms that use the generate-and-test approach to discover rules. CMDeo introduces two expansion operations and provides some properties to speed up the mining process. A lexicographic ordering is also used by CMDeo to limit some invalid expansions. In contrast to CMDeo, RuleGrowth \cite{fournier2015mining} employs a pattern-growth approach for generating SRs and achieves outstanding performance. Furthermore, ERMiner \cite{fournier2014erminer} uses an equivalence class to discover SRs. The main idea of ERMiner is to continuously merge equivalence classes to obtain longer SRs. Although ERMiner exhibits better performance in terms of running time compared to RuleGrowth, it still consumes a significant amount of memory. Besides the above-mentioned algorithms, there are also some algorithms based on SRM for other topics \cite{li2023mcor}, such as MNSR-Pretree \cite{pham2014efficient}, which aims to discover non-redundant SRs by prefix-tree, and TRuleGrowth \cite{fournier2015mining}, which uses the time period as a constraint to discover SRs from the sliding window technique.

\subsection{High-Utility Sequential Rule Mining}

High-utility sequential rule mining (HUSRM) was first defined by Zida \textit{et al.} \cite{zida2015efficient}, and the proposed algorithm is called HUSRM. Different from the traditional HUSPM task, \cite{zida2015efficient} assumes that the items in a sequence are unique. HUSRM uses the minimum utility threshold to judge whether a sequential rule is a high-utility sequential rule (HUSR). HUSRM uses a rule-growth approach to generate sequential rules and creates suitable utility tables to store key information about HUSRs. After that, building upon HUSRM, Huang \textit{et al.} \cite{huang2023us} proposed the US-Rule algorithm, which includes further REUCP and REURP strategies and an optimized upper bound calculation method. As a result, US-Rule outperforms HUSRM. The HUSRM task can be applied in other aspects, such as dealing with negative sequence data \cite{zhang2020hunsr} and repetitive items \cite{lin2023user}. In addition, Zhang \textit{et al.} \cite{zhang2024totally} proposed the totally-ordered sequential rules for utility maximization. Note that the totally-ordered sequential rules differ from the previous partially-ordered sequential rules \cite{lin2025towards}. By combining the average-utility model, the HUSRM task is extended to identify multiple objectives for high average-utility SRs from human gene expression \cite{segura2022mining}.

\section{Preliminaries and Problem Formulation}
\label{sec:preliminaries}

In this section, we first briefly introduce the concepts and definitions of high-utility sequential rule mining (HUSRM) and then formalize the HUSRM problem with the concept of increasing utility ratio.

\begin{definition}[Sequence database \cite{fournier2022pattern}]
    \rm A sequence database (\textit{SDB}) is composed of multiple sequences and can be written as \textit{SDB} = $\langle$$s_1$, $\cdots$, $s_n$$\rangle$, where $n$ is a positive number greater than 1. In the database, each sequence $s$ has its unique identifier (\textit{sid}) and contains multiple itemsets that can also be denoted as $s$ = $\langle$$I_1$, $\cdots$, $I_m$ $\rangle$ and satisfy $1 \leq m$, $m \in \mathbb{N}^+$. As for an itemset $I$, it is a set of distinct items, denoted as $I$ = $\{i_1$, $\cdots$, $i_g\}$, where $1 \leq g$ and $g \in \mathbb{N}^+$. In the paper, all items in $I$ are ordered according to the lexicographical order, i.e., $\succ_{lex}$. Each item has its corresponding internal utility and external utility, denoted as $q(i)$ and $\textit{eu}(i)$, respectively.
\end{definition}

In alignment with several previous studies \cite{zida2015efficient, huang2023us} on HUSRM, we adopt a similar concept, namely that each item in a sequence appears only once. A sample of \textit{SDB} is shown in Table \ref{tab:database}. We can see that the itemset \{$(c:4)$ $(e:2)$\} appears in $s_1$, and items $c$ and $e$ satisfy the lexicographical order $c \succ_{lex} e$. $q(c)$ = 4 and $q(e)$ = 2. Besides, the details of $eu$ are shown in Table \ref{tab:eu}, where $eu(c)$ = 4 and $eu(e)$ = 1.

\begin{table}[!ht]
	\centering
	\caption{A sample of sequence database}
	\label{tab:database}
	\begin{tabular}{cc}
		\hline
		\textbf{\textit{sid}} & \textbf{\textit{sequence}} \\
		\hline
		$s_1$ & $\langle \{(a:1)(b:4)\}, \{(c:4)(e:2)\}, \{(f:9)\},\{(d:5)\} \rangle$ \\
		$s_2$ & $\langle \{(a:2)(c:3)\}, \{(b:4)\}, \{(e:5)\},\{(f:4)\} \rangle$ \\
		$s_3$ & $\langle \{(b:2)(d:4)\}, \{(a:3)\}, \{(c:3)\},\{(e:2)\} \rangle$ \\
		$s_4$ & $\langle \{(d:3)(f:2)\}, \{(b:4)(e:1)\}, \{(g:15)\},\{(c:6)\} \rangle$ \\
		$s_5$ & $\langle \{(b:1)\}, \{(d:1)\}, \{(g:2)\} \rangle$ \\
		\hline
	\end{tabular}
\end{table}

\begin{table}[h]
	\centering
	\caption{The external utility of each item}
	\label{tab:eu}
	\begin{tabular}{cccccccc}
		\hline
		\textbf{Item} &\textbf{\textit{a}} & \textbf{\textit{b}} & \textbf{\textit{c}} & \textbf{\textit{d}} & \textbf{\textit{e}} & \textbf{\textit{f}} & \textbf{\textit{g}} \\
		\hline
		\textbf{eu} & 9 & 5 & 4 & 2 & 1 & 7 & 3\\
		\hline
	\end{tabular}
\end{table}

\begin{definition}[Sequential rule \cite{fournier2014erminer}]
    \rm The format of a sequential rule $r$ can be represented as $\{X\} \to \{Y\}$ and consists of two non-empty itemsets $X$ and $Y$, where $X$ is called antecedent, $Y$ is called consequent, and $X \cap Y$ = $\emptyset$. For the partially-ordered SRM, if a sequence $s$ with $n$ itemsets contains a sequential rule $r$, there must exist positive integers $e$ and $f$ that satisfy $X$ $\subseteq$ $\bigcup_{i=1}^{e} I_i$ ($I_i$ is the $i-$th itemset of $s$) and $Y$ $\subseteq$ $\bigcup_{j=f}^{n} I_j$, where $e$ $\textless$ $f$.
\end{definition}

\begin{definition}[The size of a sequential rule]
    \rm Given a sequential rule $\alpha$, composed of \textit{k} distinct items in the antecedent and \textit{m} distinct items in the consequent, the size of $\alpha$ is denoted as $k * m$. For another sequential rule $\beta$ with size $l * n$, only when $k$ $\leq$ $l$ and $m$ $\textless$ $n$, or $k$ $\textless$ $l$ and $m$ $\leq$ $n$, is the size of $\beta$ larger than $\alpha$.
\end{definition}

For example, the size of the sequential rule $\alpha$ = $\{a\} \to \{f\}$ is 1 $*$ 1. Given another sequential rule $\beta$ = $\{a,e\} \to \{f\}$, we can know that the size of $\beta$ is larger than $\alpha$. 

\begin{definition}[Confidence \cite{han2001prefixspan}]
    \label{def4:supAndConf}
    \rm Let $\textit{seq}(e)$ represent all sequences in the \textit{SDB} containing the element $e$, which means $\textit{seq}(e)$ = $\{s | e \subseteq s \wedge s \in \textit{SDB}\}$ and $e$ can represent a sequential rule, an itemset, or an item. Moreover, $|\textit{seq}(e)|$ is the number of sequences that contain $e$. For a sequential rule \textit{r} = $\{X\} \to \{Y\}$, the confidence of it is defined as $\textit{conf}(r)$ = $|\textit{seq}(r)| / |\textit{seq}(X)|$.
\end{definition}

In Table \ref{tab:database}, a sequential rule $r$ = $\{a\} \to \{e\}$ appears in $s_1$, $s_2$, and $s_3$, respectively. Thus, $\textit{seq}(r)$ = $\{s_1, s_2, s_3\}$, and then $|\textit{seq}(r)|$ = 3. Similarly, $|\textit{seq}(a)|$ = 3, and then the $\textit{conf}(r)$ can be calculated as $|\textit{seq}(r)|$ / $|\textit{seq}(a)|$ = 1.

\begin{definition}[Utility of item / itemset / sequential rule]
    \rm The utility of an item in \textit{SDB} is the sum of its utility in each sequence, which can be defined as $u(i)$ = $\sum_{i \in s_l \wedge s_l \in \textit{SDB}}$$q(i, s_l)$ $\times $ $eu(i)$, where $q(i, s_l)$ represents the $q(i)$ in the $l-$th sequence of \textit{SDB}. As for an itemset \textit{I}, the utility of it is the sum of the utility of items in \textit{I}, which can be denoted as $u(I)$ = $\sum_{s_l \in \textit{seq}(I)}\sum_{i \in I }$ $q(i,s_l)$ $\times$ $eu(i)$. Furthermore, the utility of a sequential rule \textit{r} = $\{X\} \to \{Y\}$ in a specific sequence $s_l$ is the sum of the itemsets $X$ and $Y$, which can be defined as $u(r, s_l)$ = $u(X)$ + $u(Y)$, where $X$ and $Y$ are both in $s_l$. Moreover, in \textit{SDB}, $u(r)$ = $\sum_{s_l \in \textit{seq}(r)} u(r, s_l)$.
\end{definition}

In Table \ref{tab:database}, since the item $a$ appears in $s_1$, $s_2$, and $s_3$, the utility of it can be calculated as $u(a)$ = (1 $\times$ 9) + (2 $\times$ 9) + (3 $\times$ 9) = 54. For the itemset $\{a,b\}$, it only appears in $s_1$ that the utility of $\{a,b\}$ is $u(\{a,b\})$ = 1 $\times$ 9 + 4 $\times$ 5 = 29. As for the sequential rule $r$ = $\{a\} \to \{e\}$ in Table \ref{tab:database} which \textit{seq}(r) = \{$s_1$, $s_2$, $s_3$\}. Then, the $u(r)$ is (1 $\times$ 9 + 2 $\times$ 1) + (2 $\times $ 9 + 5 $\times $ 1) + (3 $\times $ 9 + 2 $\times $ 1) = 63.

A sequential rule $r$ is called a HUSR if and only if \textit{minUtil} $\geq u(r)$ and \textit{minconf} $\geq \textit{conf}(r)$ where \textit{minUtil} and \textit{minconf} are predefined by the user. The \textit{minUtil} threshold is positive, and the \textit{minconf} threshold $ \in (0,1]$. The goal of HUSRM is to discover all HUSRs in \textit{SDB}. Similar to the introduction of expansion in CMDeo \cite{fournier2012cmrules}, the expansion of HUSR can also be divided into two ways: the left expansion and the right expansion. Different expansion methods have distinct properties that can be used. The expansion of the sequential rule posited by previous studies \cite{zida2015efficient, huang2023us} stipulates that a sequential rule prohibits right expansion after left expansion. The SRIU algorithm is also based on the idea of rule-growth to generate a sequential rule $r$ = $\{X\} \to \{Y\}$, but the expansion method of it is different from \cite{zida2015efficient, huang2023us}. To facilitate the discussion, we have a simple definition of the expansion of sequence rules.

\begin{definition}[Expansion of sequential rule]
    \rm For a sequential rule $r$ = \{X\} $\to$ \{Y\}, both the left expansion and the right of $r$, the new item can be added into \{X\} or \{Y\} only if it satisfies the $\succ_{lex}$ order for each item in X or Y. To avoid generating the same sequential rule again, \cite{zida2015efficient} and \cite{huang2023us} stipulate that $r$ cannot execute the right expansion after the left expansion, but the regulation in \cite{zhang2024totally} is that $r$ cannot perform the left expansion after the right expansion. To classify the two expansion methods in \cite{zida2015efficient,huang2023us} and \cite{zhang2024totally}, we call them \textit{right-left} expansion and \textit{left-right} expansion, respectively.
\end{definition}

\textbf{Problem statement}. Intuitively, the sequential rules generated based on the rule-growth approach exist in a certain relationship because the difference between the two HUSRs generated by the rule-growth approach is only an item. The first problem is: could we find out if all HUSRs satisfy an increasing utility ratio? Furthermore, the expansion order will affect the mining results. The second problem is how to select the expansion method to discover more valuable HUSRs.

For example, for the first problem, we assume the \textit{minUtil} is 10. There are two HUSRs, $r$ and $r^\prime$, where $r^\prime$ is an extension of $r$, and their utility is 50 and 40, respectively. However, the utility of $r^\prime$ goes down when compared with $r$, which means the expansion is negative. This lack of insight into the utility evolution between related rules limits their decision-making value. Consequently, practitioners cannot distinguish between truly progressive patterns and those that merely appear frequent or high-utility in isolation. As for the second problem, there is a target sequential rule $t$ = $\{a, c\}$ $\to$ $\{b,d\}$. It can be generated by $\beta$ = $\{a, c\}$ $\to$ $\{b\}$ adding the item $d$ or $\gamma$ = $\{a \}$ $\to$ $\{b,d\}$ adding the item $c$. The increasing utility ratios between the $\beta$ and $t$ or $\gamma$ and $t$ are dependent on the new item and the original utility of the sequential rule. From $\beta$ to $t$ is used for \textit{left-right} expansion, and $\gamma$ to $t$ is utilized for \textit{right-left} expansion. In these two cases, there exists one possible consequence: $t$ satisfies the increasing utility ratio for $\beta$ but not for $\gamma$. Therefore, the mining results may be different for the two expansion methods. The practical necessity for discovering rules with an increasing utility ratio is evident in real-world applications. In e-commerce, traditional HUSRM might yield a rule like $\{ \text{gamepad} \} \rightarrow \{ \text{computer game}, \text{cheap cable} \}$, where adding a low-margin accessory dilutes overall profit. In contrast, SRIU can discover value-increasing chains, such as $\{ \text{gamepad} \} \rightarrow \{ \text{computer game} \} \Rightarrow \{ \text{gamepad} \} \rightarrow \{ \text{computer game}, \text{premium subscription} \}$, guiding upselling toward higher customer lifetime value. Similarly, in healthcare analytics, while a rule like this $\{ \text{symptom A} \} \rightarrow \{ \text{diagnosis B}, \text{unnecessary test C} \}$ may be high-utility due to frequent billing, it adds cost without clinical benefit. SRIU ensures each extension contributes positively to a composite utility (e.g., diagnostic certainty minus cost), thereby supporting cost-effective care pathways. These examples underscore that the increasing utility ratio constraint is not merely technical—it enforces actionable, value-accretive evolution critical for decision-making.
		
\section{The SRIU Algorithm} \label{sec:algorithm}
		
In this section, we first introduce some key definitions to make the downward closure property valid in the mining procedure. Then, based on these new upper bounds, some corresponding pruning strategies and data structures are proposed. Finally, the SRIU algorithm is described in detail. 
		
\subsection{Definitions and upper bounds}

\begin{definition}[Utility of sequence]
    \rm The utility of a sequence is the cumulative utility of all items in the sequence, defined as $\textit{SU}(s_l)$ = $\sum_{i | i \in s_l}u(i)$ = $\sum_{i | i \in s_l} q(i,s_l)$ $\times$ $eu(i)$.
\end{definition}

For example, the utility of sequences in Table \ref{tab:database} (from $s_1$ to $s_5$) is 113, 83, 59, 110, and 13, respectively.

\begin{definition}[Sequence estimated utility, \textit{SEU} \cite{zida2015efficient}]
    \label{SEU}
    \rm In a \textit{SDB}, the sequence estimated utility (\textit{SEU}) of an item $i$ is the sum of the utility of the sequences that contain it; the \textit{SEU} of $i$ can be denoted as $\textit{SEU}(i)$ and defined as $\textit{SEU}(i)$ = $\sum_{s_l \in \textit{seq}(i)} SU(s_l)$. As for a sequential rule $r$, the \textit{SEU}$(r)$ is the \textit{SU} of the sequences in \textit{seq}($r$), which means $\textit{SEU}(r)$ = $\sum_{s_l \in \textit{seq}(r)}SU(s_l)$. 
\end{definition}

\begin{strategy}[\textit{SEU} pruning strategy, \textit{SEUP}]
    \label{strategy:SEUP}
    \rm \textit{SEUP} \cite{zida2015efficient} was proposed to prune the unpromising items and rules. For more specific, given an item $i$, if $\textit{SEU}(i)$ $\textless$ \textit{minUtil}, that $i$ can be removed from \textit{SDB}, and for a sequential rule $r$, if $\textit{SEU}(r)$ $\textless $ \textit{minUtil}, then $r$ cannot be extended further.
\end{strategy}

For example, in Table \ref{tab:database} and setting \textit{minUtil} = 200, since item $g$ appears in $s_3$ and $s_5$, then the $\textit{SEU}(g)$ = 59 + 13 = 72. In this case, $g$ is an unpromising item that can be removed. As for the sequential rule $r$ = $\{b\} \to \{f\}$, it occurs in $s_1$ and $s_2$, thus $\textit{SEU}(r)$ = 113 + 83 = 196, which means the expansion of $r$ can be terminated. 

\begin{strategy}[Confidence monotonic pruning strategy, \textit{CONFP}]
    \label{strategy:confPrune}
    \rm From Definition \ref{def4:supAndConf} we can know, with the right expansion of a sequential rule $r$ to obtain $r^\prime$, that $|\textit{seq}(r^\prime = X \to Y \cap i)|$ $\leq$ $|\textit{seq}(r = X \to Y)|$, but $|\textit{seq}(X)|$ is not changed. Therefore, the sequential rule of confidence is monotonic during the right expansion procedure. Then, if $\textit{conf}(r)$ $\leq$ \textit{minconf}, that $r$ should not be extended.
\end{strategy}

Note that \textit{right-left} expansion allows performing left expansion after executing right expansion. According to the \textit{CONFP} pruning the sequential rule ahead, the algorithm will miss some results since confidence is not monotonically decreasing in the left extension. Conversely, \textit{left-right} expansion exclusively permits right expansion after right expansion. Therefore, \textit{CONFP} can only be used for the \textit{left-right} expansion.

\begin{definition}[Item pair estimated utility pruning map]
    \rm Based on the structure REUCM \cite{huang2023us}, we proposed a new structure called item pair estimated utility pruning map (IPEUM). IPEUM does not store item pairs that are items in the same itemset when compared with REUCM. In IPEUM, the item pair $(a,b)$ is different from $(b,a)$ since they mean two sequences, respectively.
\end{definition}

For example, the item pair $(a,c)$ in Table \ref{tab:database} appears in $s_1$, $s_2$, and $s_3$. According to REUCM's definition, the $\textit{SEU}$ of $(a,c)$ in REUCM is calculated as 113 + 83 + 59 = 255. However, in IPEUM, the $\textit{SEU}$ of $(a,c)$ is calculated as 113 + 59 = 172, since the item $a$ and the item $c$ are in the same itemset in $s_2$.

\begin{theorem}
    \label{theorem:IPEUP}
    \rm Given any item pair $(a,b)$, a sequential rule $r$ = $\{a\} \to \{b\}$ and $r^\prime$ the expansion of $r$, there is $\textit{SEU}(r^\prime)$ $\leq $ $\textit{SEU}(r)$ = $\textit{IPEUM}(a,b)$.
\end{theorem}
\begin{proof}
    Given a sequential rule, $r$ = $\{X\} \to \{Y\}$, an item $i$ is used to be extended. Assume that the maximum lexicographical order items in $X$ and $Y$ are $m$ and $n$, respectively. During the rule-growth procedure, $i$ is exclusively added to the antecedent or the consequent of $r$ in an expansion. Moreover, with the expansion, $|\textit{seq}(X \cup i)|$ $\leq$ $|\textit{seq}(i)|$ or $|\textit{seq}(Y \cup i)|$ $\leq$ $|\textit{seq}(i)|$. Therefore, according to Definition \ref{SEU}, the inequalities $\textit{SEU}(\{X \cup i\}$ $\to \{Y\}) \leq \textit{SEU}(\{i\} \to \{n\})$ and $\textit{SEU}(\{X\} $ $\to $ $\{Y \cup i\})$ $\leq$ $\textit{SEU}(\{m\} \to \{i\})$ be held.
\end{proof}

\begin{strategy}[Item pair estimated utility pruning strategy, \textit{IPEUP}]
    \rm Based on the novel data structure IPEUM and Theorem \ref{theorem:IPEUP}, we propose the item pair estimated utility pruning strategy (\textit{IPEUP}) to reduce the search space. Given a sequential rule $r$ = $\{X\} \to \{Y\}$, and $m$ and $n$ are the largest items according to $\succ_{lex}$ in $X$ and $Y$, respectively. For an item $i$, if \textit{IPEUM}$(m,i)$ $\textless$ \textit{minUtil}, $r$ will not be used to execute right expansion on $i$, and if \textit{IPEUM}$(i,n)$ $\textless $ \textit{minUtil}, $r$ will not be used to execute left expansion on $i$.
\end{strategy}

\begin{definition}[Utility of expansion]
    \label{def:UofExpansion}
    \rm The total utility of all available items in the sequence $s$ only for a sequential rule $r$ to execute the right expansion is denoted as \textit{UR}$(r,s)$. Similarly, the left expansion also has \textit{UL}$(r,s)$. The total utility of all items in the sequence $s$ that is available for $r$ to perform left and right expansions is denoted as \textit{ULR}$(r,s)$.
\end{definition}

To determine the optimal expansion strategy to employ for attaining more valuable outcomes, we propose the $e$-\textit{index} and the definition is below. 

\begin{definition}[e-index]
    \rm The average utility of items that can be used for left expansion for a sequential rule $r$ in a sequence $s$ is defined as $\textit{AvgULeft}(r,s)$ = $(\textit{UL}(r,s) + \textit{ULR}(r,s))$ / $\textit{distinct}(i_L)$, where $\textit{distinct}(i_L)$ means the number of distinct items for $r$ to execute left expansion. Similarly, the right expansion also has $\textit{AvgURight}(r,s)$ = $(\textit{UR}(r,s) + \textit{ULR}(r,s))$ / $\textit{distinct}(i_R)$. We use the index $e$-\textit{index} to measure whether a $1*1$ sequential rule should be extended by \textit{left-right} expansion or \textit{right-left} expansion. 
    $$ e(r,s)=\left\{
	\begin{aligned}
	1 & & \textit{AvgULeft}(r,s) \geq \textit{AvgURight}(r,s)
        \\
	  - 1 & & \textit{AvgULeft}(r,s) \textless \textit{AvgURight}(r,s)
	\end{aligned}
	\right.
	$$
    $e(r, \textit{SDB}) = \sum_{s | s\in \textit{seq}(r)}{e(r,s)}$. Then, if $e(r, \textit{SDB}) \geq 0$ that $r$ adopts \textit{right-left} expansion; otherwise, $r$ adopts \textit{left-right} expansion.
\end{definition}

Consider the sequential rule $r$ = $\{a\} \to \{d\}$ in $s_1$ in Table \ref{tab:database}. In $s_1$, the items $b$ and $c$ are only used to perform left expansion, and then \textit{UL}$(r, s_1)$ = 4 $\times$ 5 + 4 $\times$ 4 = 36. The items $e$ and $f$ are available for both the left and right expansions, thus \textit{ULR}$(r,s_1)$ = 2 $\times$ 1 + 9 $\times$ 7 = 65. Since no item in $s_1$ only applies to the right expansion of $r$, \textit{UR} = 0. For left expansion, \textit{distinct}$(i_L)$ = 2 + 2 = 4, and as for right expansion, \textit{distinct}$(i_R)$ = 2. Thus, \textit{AvgULeft}$(r,s)$ = (36 + 65) / 4 = 25.25 and \textit{AvgURight}$(r,s)$ = 65 / 2 = 32.5. Then, $e(r, \textit{SDB})$ = $e(r, s_1)$ = -1, and then SRIU adopts \textit{left-right} expansion for $r$.

In order to reduce the search space of SRIU, there are two sets of upper bounds and strategies for the \textit{left-right} expansion and the \textit{right-left} expansion, respectively.

\begin{upper bound}[Estimated utility of right expansion]
    \label{EURE}
    \rm  The estimated utility of right expansion (EURE) of a sequential rule $r$ in the sequence $s$, denoted as \textit{EURE}$(r, s)$, is defined as:
    $$ \textit{EURE}(r,s)=\left\{
	\begin{aligned}
	u(r,s) + \textit{AllUR}(r,s) &  & \textit{AllUR}(r,s) \textgreater 0 \\
	0 & & otherwise
	\end{aligned}
	\right.
	$$
\end{upper bound}

Here, \textit{AllUR}$(r,s)$ represents all the utilities of items that can be used for right expansion for $r$, that is, \textit{AllUR}$(r,s)$ = \textit{UR}$(r,s)$ + \textit{ULR}$(r,s)$. The estimated utility of the right expansion of a sequential rule $r$ in a \textit{SDB} can be defined as $\textit{EURE}(r,\textit{SDB})$ = $\sum_{s_l \in \textit{seq}(r)}{\textit{EURE}(r,s_l)}$.

For example, there exists a sequential rule $r$ = $\{b\} \to \{c\}$ in Table \ref{tab:database} which \textit{seq}$(r)$ = \{$s_1$, $s_3$, $s_4$\}. In $s_1$, since there is no item between items $b$ and $c$, \textit{ULR}$(r, s_1)$ = 0. However, both items $e$, $f$, and $d$ can be used to perform the right expansion of $r$; thus, \textit{UR}$(r, s_1)$ = $u(e, s_1)$ + $u(f, s_1)$ + $u(d, s_1)$ = 2 + 63 + 10 = 75. Then \textit{EURE}$(r, s_1)$ = $u(r, s_1)$ + \textit{AllUR}$(r,s_1)$ = 36 + 0 + 75 = 111. In the same way, \textit{EURE}$(r, s_3)$ and \textit{EURE}$(r, s_4)$ are equal to 22 and 90, respectively. Therefore, \textit{EURE}$(r)$ = 111 + 22 + 90 = 223.

\begin{theorem}
    \label{theorem:EURE}
    \rm Given a sequential rule $r$ = $\{X\} \to \{Y\}$ and its right expansion $r^\prime$ = $\{X\} \to \{Y \cup a\}$, there is always satisfy $u(r^\prime) $ $\leq \textit{EURE}(r^\prime)$ $\leq$ $\textit{EURE}(r)$.
\end{theorem}
\begin{proof}
    \label{proof:EURE}
    \rm According to upper bound \ref{EURE}, it is clearly that $u(r^\prime)$ $\leq$ $\textit{EURE}(r^\prime)$. Since the item $a$ can be used for right expansion for $r$, $u(a,s)$ is part of $u(\textit{UR})(r, s)$ + $u(\textit{ULR})(r,s)$ = \textit{AllUR}$(r,s)$. Furthermore, the extension of $r$ is only feasible when the item fulfills the condition of $\succ_{lex}$, which means that when adding a new item to $r$, the size of the set of items that can be the right expansion for $r^\prime$ is no larger than $r$.
    \begin{tabbing}
        $\because$ $u(a,s)$ + \textit{AllUR}$(r^\prime, s)$ $\leq$ \textit{AllUR}$(r,s)$,
    \end{tabbing}
    \begin{tabbing}
	     $\therefore$ \textit{EURE}$(r^\prime, s)$ \= 
	    = $u(r^\prime, s)$ + \textit{AllUR}$(r^\prime, s)$ \\
	    \> = $u(r, s) + u(a,s)$ + \textit{AllUR}$(r^\prime, s)$ \\
	    \> $\leq$ $u(r, s)$ + \textit{AllUR}$(r,s)$ \\
	    \> = \textit{EURE}$(r, s)$. 
    \end{tabbing}   
    Moreover, it is easy to get \textit{seq}$(r^\prime)$ $\in$ \textit{seq}$(r)$. Therefore, $u(r^\prime)$ $\leq$ $\sum_{s_l \in \textit{seq}(r^\prime)}{\textit{EURE}(r^\prime, s_l)}$ $\leq$ $\sum_{s_l \in \textit{seq}(r)}{\textit{EURE}(r, s_l)}$.
\end{proof}

\begin{strategy}[Estimated utility of right expansion pruning strategy]
    \label{strategy:EURE}
    \rm According to Theorem \ref{theorem:EURE}, during the right expansion procedure of the \textit{left-right} expansion of a sequential rule $r$, the SRIU algorithm proposes the estimated utility of the right expansion pruning strategy (EUREP), which can safely terminate the expansion of $r$ when \textit{EURE}$(r)$ $\textless$ \textit{minUtil}.
\end{strategy}

\begin{upper bound}[Estimated utility of left expansion]
    \label{EULE}
    \rm  The estimated utility of left expansion (EULE) of a sequential rule $r$ in the sequence $s$, denoted as \textit{EULE}$(r, s)$, is defined as:
    $$ \textit{EULE}(r,s) = \left\{
	\begin{aligned}
	u(r,s) + \textit{UExtend}(r,s) & & \textit{UExtend}(r,s) \textgreater 0 \\
	0 & & otherwise
	\end{aligned}
	\right.
	$$
\end{upper bound}

Here, \textit{UExtend}$(r,s)$ represents the aggregate utility of all available items for $r$ to execute left or right expansions, which means \textit{UExtend}$(r,s)$ = \textit{UL}$(r,s)$ + \textit{ULR}$(r,s)$ + \textit{UR}$(r,s)$. Different from upper bound \ref{EURE}, since the algorithm can execute right expansion after left expansion according to the \textit{left-right} expansion, that algorithm must take the \textit{UR} of $r$ into account during the left expansion. The estimated utility of the left expansion of a sequential rule $r$ in an \textit{SDB} is defined as: \textit{EULE}$(r)$ = $\sum_{s_l \in \textit{seq}(r)}{\textit{EULE}(r,s_l)}$.

We continue to consider the sequential rule $r$ = $\{b\} \to \{c\}$ in Table \ref{tab:database} which \textit{seq}$(r)$ = \{$s_1$, $s_3$, $s_4$\} as an example. In $s_4$, items $e$ and $g$ are between $b$ and $c$, both of which are available for $r$ to perform left and right expansions. Therefore, \textit{ULR}$(r, s_4)$ = $u(e, s_4)$ + $u(e, s_4)$ = 1 + 45 = 46. In $s_4$, $r$ can perform left expansion by adding the items $d$ and $f$, thus \textit{UL}$(r, s_4)$ = $u(d, s_4)$ + $u(f, s_4)$ = 6 + 14 = 20. As there is no item for $r$ to right expansion, \textit{EULE}$(r, s_4)$ = $u(r, s_4)$ + \textit{UL}$(r, s_4)$ + \textit{ULR}$(r, s_4)$ + \textit{UR}$(r, s_4)$ = 44 + 20 + 46 + 0 = 110. In the same way, \textit{EULE}$(r, s_1)$ and \textit{EULE}$(r, s_3)$ are equal to 111 and 22, respectively. As a result, \textit{EULE}$(r) $ = 111 + 22 + 110 = 243.

\begin{theorem}
    \label{theorem:EULE}
    \rm Given a sequential rule $r$ = $\{X\} \to \{Y\}$ and its left expanded rule $r^\prime$ = $\{X \cup a\} \to \{Y\}$, there is always $u(r^\prime)$ $\leq$ \textit{EULE}$(r^\prime)$ $\leq$ \textit{EULE}$(r)$.
\end{theorem}

\begin{proof}
    \label{proof:EULE}
    \rm According to the upper bound \ref{EULE}, it is clear that $u(r^\prime)$ $\leq$ \textit{EULE}$(r^\prime)$. After adding $a$ into $r$ to generate $r^\prime$, the items previously deemed suitable for left expansion in $r$ may no longer be compatible with $r^\prime$. However, the items that can be used for expansion as before have not changed. We have $u(a,s)$ + \textit{UL}$(r^\prime, s)$ + \textit{ULR}$(r^\prime,s)$ + \textit{UR}$(r^\prime, s)$ $\leq$ \textit{UL}$(r,s)$ + \textit{ULR}$(r,s)$ + \textit{UR}$(r, s)$. 
    \begin{tabbing}
	     $\therefore$ $\textit{EULE}(r^\prime, s)$ \= 
	    = $u(r^\prime, s)$ + \textit{UExtend}$(r^\prime, s)$ \\
	    \> = $u(r, s)$ + $u(a,s)$ + \textit{UExtend}$(r^\prime, s)$ \\
	    \> $\leq$ $ u(r, s)$ + \textit{UExtend}$(r, s)$ \\
	    \> = \textit{EULE}$(r, s)$. 
    \end{tabbing}    
    Moreover, it is easy to know \textit{seq}$(r^\prime)$ $\in$ \textit{seq}$(r)$. Therefore, $u(r^\prime)$ $\leq$ $\sum_{s_l \in \textit{seq}(r^\prime)}{\textit{EULE}(r^\prime, s_l)} \leq \sum_{s_l \in \textit{seq}(r)}{\textit{EULE}(r, s_l)}$.
\end{proof}

\begin{strategy}[Estimated utility of left expansion pruning strategy]
    \label{strategy:EULE}
    \rm According to Theorem \ref{theorem:EULE}, during the left expansion procedure of the \textit{left-right} expansion of a sequential rule $r$, SRIU proposes the estimated utility of the left expansion pruning strategy (EULEP), which can safely terminate the expansion of $r$ when \textit{EULE}$(r)$ $\textless $ \textit{minUtil}.
\end{strategy}

Upper bounds \ref{EURE} and \ref{EULE} and the accompanying strategies \ref{strategy:EURE} and \ref{strategy:EULE} are applicable for \textit{left-right} expansion. In contrast, the upper bounds and pruning strategies employed in \textit{right-left} expansion differ from those in \textit{left-right} expansion, as outlined in the prior work \cite{huang2023us}. In the subsequent discussion, we provide concise definitions of the related concepts.

\begin{upper bound}[Left expansion estimated utility]
    \label{LEEU}
    \rm The left expansion estimated utility (LEEU) of a sequential rule $r$ in a sequence is defined as: 
    $$ \textit{LEEU}(r,s)=\left\{
	\begin{aligned}
	u(r,s) + \textit{AllUL}(r,s)&  & \textit{AllUL}(r,s) \textgreater 0 \\
	0 & & otherwise
	\end{aligned}
	\right.
	$$
    Here, \textit{AllUL} is similar to \textit{AllUR}, which denotes the total utility of items that can be used for left expansion for $r$: \textit{AllUL}$(r,s)$ = \textit{UL}$(r,s)$ + \textit{ULR}$(r,s)$. The left expansion estimated utility of a sequential rule $r$ in an \textit{SDB} can be defined as $\textit{LEEU}(r)$ = $\sum_{s_l \in \textit{seq}(r)}{\textit{LEEU}(r,s_l)}$.
\end{upper bound}

\begin{strategy}[LEEU pruning strategy: LEEUP]
    \label{strategy:LEEU}
    \rm The downward closure of LEEU has been proven in \cite{huang2023us}; then, SRIU can safely terminate the left expansion of the \textit{right-left} expansion of $r$ when $\textit{LEEU}(r) \textless$ \textit{minUtil}.  
\end{strategy}

\begin{upper bound}[Right expansion estimated utility]
    \label{REEU}
    \rm \rm The right expansion estimated utility (REEU) of a sequential rule $r$ in a sequence is defined as:
    $$ \textit{REEU}(r,s)=\left\{
	\begin{aligned}
	u(r,s) + \textit{UExtend}(r,s) & & \textit{UExtend}(r,s) \textgreater 0 \\
	0 & & otherwise
	\end{aligned}
	\right.
	$$
Here, \textit{UExtend} still indicates the total utility of items that can be used for the left or right expansion for $r$. The right expansion estimated utility of a sequential rule $r$ in an \textit{SDB} can be defined as \textit{REEU}$(r)$ = $\sum_{s_l \in \textit{seq}(r)}{\textit{REEU}(r,s_l)}$.
\end{upper bound}

\begin{strategy}[REEU pruning strategy: REEUP]
    \label{strategy:REEU}
    \rm The downward closure of REEU has been proven in \cite{huang2023us}; then, SRIU can safely terminate the right expansion of the \textit{right-left} expansion of $r$ when \textit{REEU}$(r)$ $\textless$ \textit{minUtil}. 
\end{strategy}

Although a lot of methods \cite{gan2019survey,telikani2020survey} were proposed to evaluate association rules like \textit{lift}, it is not a null invariance measure (the sequence does not contain any investigation items), which means the number of sequences in the database will affect the value of the metric. Other methods, such as \textit{all\_conf}, \textit{max\_conf}, \textit{cosine}, and \textit{Kulc}, do not consider the implication relation between a rule's antecedent and consequent. Study \cite{luna2018optimization} strictly analyzes the relationship between several rule measures. Based on this, we evaluate sequence rules with confidence and \textit{conviction}. Although \textit{conviction} is not a null invariance measure, it represents a rule's degree of implication and is consistent with the order relationship emphasized by the sequential rule. The formulation of \textit{conviction} is defined as \textit{conviction}$(X \to Y)$ = $(P(X)$ $\times$ $P(\neg Y))$ / $ P(X \neg Y)$. To further simplify the calculation, \textit{conviction}$(X \to Y)$ can be rewritten as $(1 - P(Y))$ / $1-\textit{conf}(X \to Y))$, where P(X) represents the probability of the occurrence of $X$ in the database. A higher value of \textit{conviction} indicates a stronger association between the sequential rule's antecedent and consequent.

\subsection{Data structures}

We design two data structures, \textit{UTable} and \textit{PUTable}, to achieve the pruning strategies based on the abovementioned upper bounds. Meanwhile, to calculate the $e$-index of each $1 * 1$ sequential rule and avoid carrying redundant information during the mining procedure, we also design a temp data structure \textit{TempTable} to store relevant information of $1 * 1$ sequential rule, and \textit{TempTable} can easily transform to \textit{UTable}. The details are described below.

\begin{definition}[TempTable]
  \rm  TempTable is designed for the $1 * 1$ sequential rule $r$. As shown in Fig. \ref{fig:TempTable}, each TempTable contains the following auxiliary information: 1) the total utility (\textit{tu}) of $r$ in the \textit{SDB}; 2) the sum of estimated utility (\textit{sumEU}) of $r$; 3) the $e$-index of $r$; and 4) a list of \textit{Temp-Elements}. According to the \textit{sumEU}, SRIU can prune the $1 * 1$ sequential rule at the beginning. The value of $e$-index is used to determine whether $r$ adopts the \textit{left-right} expansion or the \textit{right-left} expansion. As for \textit{Temp-Element}, which records the related information of a sequence containing $r$, the related information of $r$ in the sequence $s$ can be divided into six types, and they can be represented as a tuple like (\textit{SID}, \textit{Utility}, \textit{EU}, \textit{Positions}, \textit{Items}, \textit{e}). The sequence identifier of $s$, which contains $r$, is \textit{SID}. \textit{Utility} is the utility of $r$ in $s$. \textit{EU} is composed of three parts and defined as \textit{EU} = (\textit{UL}$(r,s)$, \textit{ULR}$(r,s)$, \textit{UR}$(r,s)$). \textit{Positions} is a 2-turple ($\alpha$, $\beta$) that consists of the position of the last item in the antecedent and consequence of $r$. As for \textit{Items}, which is used to auxiliary calculate the \textit{e}, \textit{Items} is a 2-turple (\textit{IL}, \textit{IR}) that is made up of the number of items in $s$ that can be used to perform left expansion and right expansion.
\end{definition}

Taking the sequential rule $r$ = $\{a\} \to \{d\}$ in Table \ref{tab:database} as an example. Since \textit{seq}$(r)$ = $\{s_1\}$, there is only one \textit{Temp-Element} in the \textit{TempTable} for $r$. In $s_1$, the positions of $a$ and $d$ are 1 and 3, respectively. Therefore, $\alpha$ = 1 and $\beta$ = 3, and it is easy to calculate \textit{Utility}$(r)$ = 19. Moreover, according to Definition \ref{def:UofExpansion}, \textit{UL}, \textit{ULR}, and \textit{UR} are equal to 36, 65, and 0, respectively. Meanwhile, the value of \textit{IL} is 4 since 4 items can be used to perform left expansion for $r$ in $s_1$, and then \textit{IR} = 2. Finally, based on the definition of $e$-index, $e(r,s_1)$ = -1, and the $e$-index of $r$ is equal to -1. As there is only one \textit{Temp-Element} in the \textit{TempTable} for $r$, \textit{tu} = \textit{Utility}$(r,s_1)$ = 19 and \textit{sumEU} = 19 + 36 + 65 + 0 = 120. The result is shown in Fig. \ref{fig:TempTable}.

\begin{figure}[!htbp]
    \centering
    \includegraphics[trim=110 190 100 170,clip,scale=0.33]{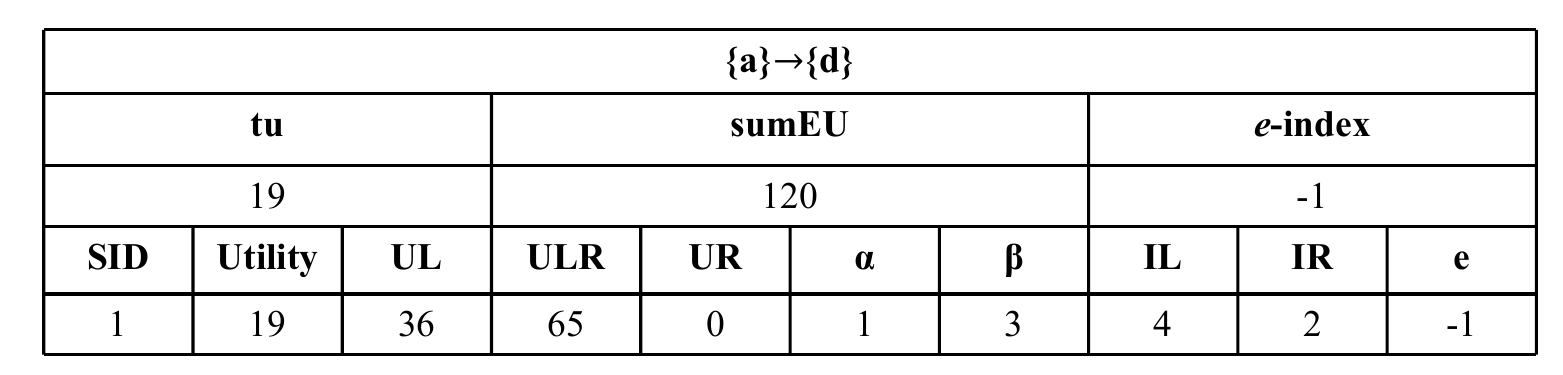}
    \caption{The TempTable w.r.t. $\{a\} \to \{d\}$.}
    \label{fig:TempTable}
\end{figure}

Once it obtains the $e$-index, SRIU can adopt the specific expansion approach. Note that, since the upper bounds of the \textit{left-right} expansion and the \textit{right-left} expansion are different, we can observe that \textit{EULE}$(r)$ and \textit{REEU}$(r)$ are expressively equivalent. Therefore, we can use a new variable to give a uniform representation of \textit{EULE}$(r)$ and \textit{REEU}$(r)$, \textit{UBTotal}. Furthermore, there is a slight difference between \textit{EURE}$(r)$ and \textit{LEEU}$(r)$ that makes it inappropriate to use two different data structures to store \textit{EURE}$(r)$ and \textit{LEEU}$(r)$, respectively. Then, in SRIU, it adopts a boolean variable, \textit{flag}, to determine whether to use \textit{EURE}$(r)$ or \textit{LEEU}$(r)$.

\begin{definition}[UElement and UTable of a HUSR]
    \rm Similar to TempTable, UTable differs in that it does not require the storage of the $e$-index of $r$. The estimated utility of $r$ is divided into two parts, \textit{UBTotal} and \textit{UBPart}, which are the sum of \textit{UBTotal} and \textit{UBPart} in elements (\textit{UElement}), respectively. \textit{UElement} is a 6-tuple consisting of (\textit{SID}, \textit{Utility}, \textit{EU}, \textit{Positions}, \textit{UBs}, \textit{flag}), where \textit{UBs} and \textit{flag} are novel elements in comparison to \textit{Temp-Element}. \textit{UBs} is a 2-tuple (\textit{UBTotal}, \textit{UBPart}), in which \textit{UBTotal} represents $\textit{EULE}(r)$ or $\textit{REEU}(r)$. As for \textit{UBPart}, when \textit{flag} = 1, then \textit{UBPart} = \textit{EURE}$(r)$, and if \textit{flag} = 0, then \textit{UBPart} = \textit{LEEU}$(r)$.
\end{definition}

\begin{figure}[!htbp]
    \centering
    \includegraphics[trim=70 160 100 170,clip,scale=0.32]{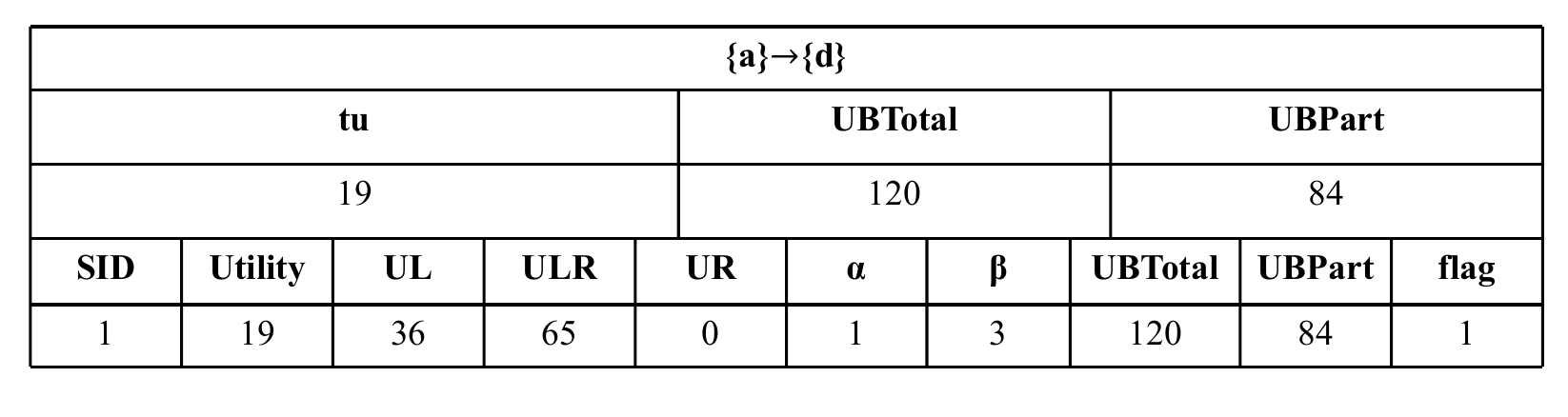}
    \caption{The TempTable w.r.t. $\{a\} \to \{d\}$.}
    \label{fig:UTable}
\end{figure}

For example, since we calculate the $e$-index of $r$ = $\{a\} \to \{d\}$ is equal to -1 before that, SRIU adopts the \textit{left-right} expansion to extend $r$. Therefore, \textit{flag} = 1 and \textit{UBPart}$(r,s_1)$ = \textit{EURE}$(r)$ = \textit{Utility}$(r,s_1)$ + \textit{ULR}$(r,s_1)$ + \textit{UR}$(r,s_1)$ = 84. Then, the \textit{UTable} of $r$ is shown in Fig. \ref{fig:UTable}.

\begin{definition}[PUTable]
    \rm Since we limit the expansion, PUTable is a more compact structure than UTable. For example, the right expansion cannot be performed after a left expansion during the \textit{right-left} expansion, and it is the same reason during the \textit{left-right} expansion. Therefore, PUTable no longer stores \textit{Position} and \textit{UBTotal}. Moreover, if \textit{flag} = 1, we save \textit{UR}, and if \textit{flag} = 0, we save \textit{UL} because it is enough to calculate the \textit{UBPart} from the remaining information.
\end{definition}

For example, the sequential rule $r$ = $\{a\} \to \{d\}$ generates the $r^\prime$ = $\{a\} \to \{d,e\}$ by the right expansion. Since $r$ is used in the \textit{left-right} expansion, the subsequent expansion of $r^\prime$ does not need to consider the left expansion. Therefore, SRIU only stores the \textit{tu} and \textit{UBPart} in \textit{PUTable} of $r^\prime$. Additionally, as \textit{flag} = 1, that \textit{PUTable} only keeps the \textit{UR}. The \textit{PUTable} of $r^\prime$ is shown in Fig. \ref{fig:PUTable}.

\begin{figure}[!htbp]
    \centering
    \includegraphics[trim=110 250 100 120,clip,scale=0.33]{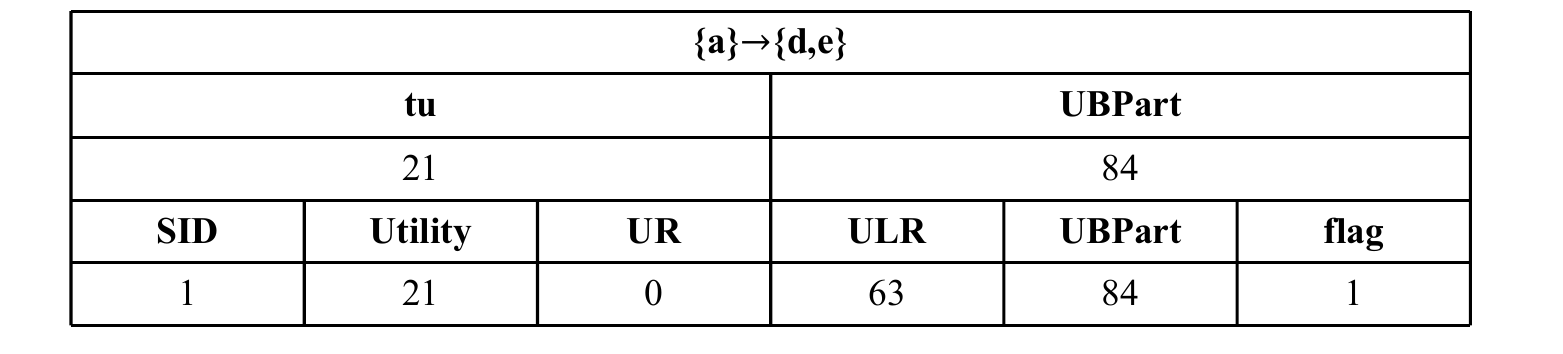}
    \caption{The PUTable w.r.t. $\{a\} \to \{d,e\}$.}
    \label{fig:PUTable}
\end{figure}

\subsection{Additional optimization}

The SRIU algorithm also utilizes two additional optimizations to improve its performance further. The first one compresses the storage space for each item's sequence identifiers. The essential operations in SRIU involve the recording of the \textit{sid} associated with each item and the computation of intersections among items. However, Bitmap is not suitable for sparse storage problems. Thus, SRIU adopts the Roaring bitmap structure \cite{chambi2016better}, which partitions integer types into high and low 16-bit segments, eliminating the redundancy typically encountered in Bitmap. This approach ensures that the Roaring Bitmap achieves high efficiency in executing intersection and union operations while minimizing memory usage. We can observe that if the antecedent or consequent adds the largest item, this sequential rule can not continue to perform left expansion or right expansion, respectively. Thus, SRIU records the largest item in each sequence in a map \textit{MaxIS} and terminates the corresponding expansion when it meets the largest item.
		
\subsection{The SRIU Algorithm}

Based on the discussions mentioned above, we propose an algorithm to discover high-utility sequential rules with an increasing utility ratio, called SRIU. SRIU generates all $1 * 1$ sequential rules. According to the \textit{sumEU} and $e$-index, it determines whether they should be expanded and adopts which expansion method. Note that the \textit{left-right} expansion and the \textit{right-left} expansion use two sets of different functions to generate all HUSRs with an increasing utility ratio. The disparities between the \textit{left-right} expansion and \textit{right-left} expansion methods pertain primarily to the sequence of expansion and the involvement of \textit{CONFP}. It is noteworthy that the \textit{right-left} expansion, along with its procedural details, has been exhaustively expounded upon in the US-Rule algorithm \cite{huang2023us}. Hence, in this section, our focus is exclusively directed toward elucidating the nuances of the \textit{left-right} expansion approach. As for the \textit{right-left} expansion, the differences are shown in $\textbf{[]}$. The pseudocode for SRIU is shown in Algorithm \ref{algo:SRIU_algorithm}.

\begin{algorithm}[!h]
\footnotesize
	\caption{The SRIU algorithm}
	\label{algo:SRIU_algorithm}
	\LinesNumbered
	\KwIn{\textit{SDB}, \textit{minUtil}, \textit{minconf}, \textit{minRatio}.}
	\KwOut{All HUSRs with increasing utility ratio.}
        initialize a Hash map $M$ to store each item $i$ in \textit{SDB} and $\textit{SEU}(i)$;
        \While {$\exists$ $i$ $\in$ $M$ and $\textit{SEU}(i)$ $\textless$ \textit{minUtil}} {
		remove all unpromising items from $M$ and in \textit{SDB}; \\
		update the SEU of all items in $M$; 
	}
        scan \textit{SDB} again to build \textit{IPEUM}, $R$ and \textit{MaxIS}; \\
        remove all item pairs in \textit{IPEUM} which estimated utility less than \textit{minUtil};\\
	remove all unpromising sequential rules from $R$; \\ 
        build the TempTable for each sequential rule $r$ in $R$;
        
	\For{$r$ $\in$ $R$}{
            calculate the \textit{conf}$(r)$;\\
	    \If{\rm \textit{tu}(r) $\ge$ \textit{minUtil} and \textit{conf}$(r)$ $\ge$ \textit{minconf}}{
		        save $r$;
		}
	    
	    \If{\textit{sumEU} $\ge$ \textit{minUtil} $\times$ \textit{minRatio}}{              
		        \If{$e$-index $\textless$ 0}{ 
                    convert the TempTable to the UTable where $\textit{flag} = 1$; \qquad(\textbf{left-right} expansion);\\
                    call \textbf{DoubleExpansion}($r$, \textit{SDB}, \textit{IPEUM}, \textit{minUtil}, \textit{minconf}, \textit{minRatio});\\
                    \If{UTable.\textit{UBPart} $\ge$ \textit{minUtil} and \textit{conf}$(r)$ $\ge$ \textit{minconf}}{
                        call \textbf{SingleExpansion}($r$, \textit{SDB}, \textit{IPEUM}, \textit{minUtil}, \textit{minconf}, \textit{minRatio});
                    }
                }\Else{
                    convert the TempTable to the UTable where $\textit{flag} = 0$;\qquad(\textbf{right-left} expansion);\\
                    call \textbf{DoubleExpansion}(($r$, \textit{SDB}, \textit{IPEUM}, \textit{minUtil}, \textit{minconf}, \textit{minRatio});\\
                    \If{UTable.\textit{UBPart} $\ge$ \textit{minUtil}}{
                        call \textbf{SingleExpansion}($r$, \textit{SDB}, \textit{IPEUM}, \textit{minUtil}, \textit{minconf}, \textit{minRatio});
                    }
                }                
		}  	   	   
	}	
\end{algorithm}

SRIU takes a sequence database \textit{SDB}, a minimum utility threshold \textit{minUtil}, a minimum confidence threshold \textit{minconf}, and a minimum increasing utility ratio \textit{minRatio} as its input. All of these thresholds are predefined by the user, and then it outputs all HUSRs satisfying the increasing utility ratio. The initial step of the SRIU algorithm involves scanning the \textit{SDB} to extract all items and compute their \textit{SEU}. Simultaneously, these key-value pairs are stored in a hash map denoted as $M$ (line 1). According to \textit{SEUP}, SRIU constantly removes the unpromising items from $M$ and \textit{SDB} and updates the \textit{SEU} of the remaining items in $M$ until there are no unpromising items in \textit{SDB} (lines 2-5). SRIU scans \textit{SDB} again, generates all $1 * 1$ sequential rules, and calculates the \textit{SEU} and \textit{seq}$(r)$ of them, then they are stored in $R$, which is a set. What's more, SRIU initializes \textit{IPEUM} and records the largest item (according to $\succ_{lex}$ order) in each sequence in \textit{MaxIS} (line 6). According to \textit{IPEUP} and \textit{SEUP}, SRIU removes the item pairs in \textit{IPEUM} whose estimated utility is less than \textit{minUtil} and removes all sequential rules for which $\textit{SEU}(r)$ is less than \textit{minUtil} in $R$ (lines 7-8). Then, build the TempTable for all promising sequential rules $r$ in $R$. Meanwhile, calculate the \textit{tu}, \textit{sumEU}, and $e$-index of $r$ (line 9). Subsequently, for each $r$, it determined \textit{conf}(r) if $r$ is a high-utility sequential rule: its \textit{tu} and \textit{conf} values meet or exceed the predefined thresholds of \textit{minUtil} and \textit{minconf}, respectively (lines 11-14). Next, according to \textit{EULEP} and \textit{REEUP}, $r$ can be extended depending on whether its \textit{sumEU} is larger than \textit{minUtil} (line 15). If $r$ can be extended, we need to select a specific expansion method based on the $e$-index. SRIU performs \textit{left-right} expansion when $e$-index is less than 0 (lines 16-22), else executes \textit{right-left} expansion (lines 23-29). SRIU uses the variable \textit{flag} to mark the algorithm's turn to which expansion strategy and adopts the corresponding pruning strategies. At first, SRIU converts the TempTable to the UTable after selecting the specific expansion method (lines 17 and 24). Note that \textit{CONFP} can only be used in \textit{left-right} expansion (line 19).

\begin{algorithm}[!h]
\footnotesize
	\caption{The DoubleExpansion Procedure}
	\label{algo:leftProcedure}
	\LinesNumbered
	\KwIn{$r$, \textit{SDB}, \textit{IPEUM}, \textit{minUtil}, \textit{minconf}, \textit{minRatio}.}
	\KwOut{All HUSRs with increasing utility ratio.}

        initial \textit{rSet} $\leftarrow$ $\emptyset$ and obtain the largest item in the consequent of $r$: $n$; \quad[obtain the largest item in the antecedent of $r$: $m$]\\
        calculate the \textit{riseU} = $u(r)$ $\times $ \textit{minRatio};\\
	\For{\rm  sequence $s \in \textit{seq}(r)$ }{
	       \If{\textit{UExtend}$(r,s)$ = 0}{
                    continue;
                }
	
     \For{$i \in s$ $\wedge$ $i \in$ \textit{LE} \quad[$i \in s \wedge i \in \textit{RE}$]}{ 
	        
                 \If{\textit{IPEUM}(i,n) \quad[\textit{IPEUM}(m,i)] $\textless$ \textit{max}(\textit{minUtil}, \textit{riseU})}{
                    continue;
                }

                $t$ $\leftarrow$ (\{$X$\} $\cup$ $i \to$ \{$Y$\}); \quad[$t$ $\leftarrow$ (\{$X$\} $\to$ \{$Y$\} $\cup$ $i$)];\\
                create UTable for $t$ where \textit{Flag} = 1; \quad[\textit{Flag} = 0]; \\
                \If{\rm  $i$ equal to \textit{MaxIS}$(s)$}{
                    \textit{UL}$(r,s)$ = 0 and \textit{UR}$(r,s)$ += \textit{ULR}$(r,s)$; 
                    \quad[\textit{UR}$(r,s)$ = 0 and \textit{UL}$(r,s)$ += \textit{ULR}$(r,s)$;]\\
                }
                add $t$ to the \textit{rSet} and update the UTable of $t$;\\	        
	    }	    
	}
	
	\For{\rm $r \in$ rSet}{
            calculate the \textit{conf}$(r)$ from its UTable;\\
	    \If{\rm  \textit{tu}(r) $\ge$ max(\textit{minUtil}, \textit{riseU}) and \textit{conf}$(r)$ $\ge$ \textit{minconf}}{
		        save $r$;
		}
          \If{\textit{UTable}.\textit{UBTotal} $\ge$ max(\textit{minUtil}, \textit{riseU})  }{
		        call \textbf{DoubleExpansion}(($r$, \textit{SDB}, \textit{IPEUM}, \textit{minUtil}, \textit{minconf}, \textit{minRatio});\\
		}
  
	    \If{\rm  \textit{UTable}.\textit{UBPart} $\ge$ max(\textit{minUtil}, \textit{riseU}) and \textit{conf}$(t)$ $\ge$ \textit{minconf}(\textit{UTable}.\textit{UBPart} $\ge$ max(\textit{minUtil}, \textit{riseU}))}{
                 call \textbf{SingleExpansion}($r$, \textit{SDB}, \textit{IPEUM}, \textit{minUtil}, \textit{minconf}, \textit{minRatio});
            }
        }
\end{algorithm}

\begin{algorithm}[h]
\footnotesize
	\caption{The SingleExpansion Procedure}
	\label{algo:rightProcedure}
	\LinesNumbered
	\KwIn{$r$, \textit{SDB}, \textit{IPEUM}, \textit{minUtil}, \textit{minconf}, \textit{minRatio}.}
	\KwOut{All HUSRs with increasing utility ratio.}

       initial \textit{rSet} $\leftarrow$ $\emptyset$ and obtain the largest item in the antecedent of $r$;\\
        calculate the \textit{riseU} = $u(r)$ $\times$ \textit{minRatio};\\
	\For{\rm sequence $s \in \textit{seq}(r)$ }{
	       \If{\textit{UExtend}$(r, s)$ = 0}  {
                    continue;
                }	    
	    \For{\rm $i \in s \wedge i \in \textit{UR}$ \quad[$i \in s \wedge i \in \textit{UR}$]}{	      
                 \If{\textit{IPEUM}(m,i) \quad[\textit{IPEUM}(i,n)] $\textless$ \textit{minUtil}}{
                    continue;
                }
                $t$ $\leftarrow$ (\{X\} $\to$ \{Y\} $\cup$ $i$); \quad[$t$ $\leftarrow$ (\{X\} $\cup$ $i$ $\to$ \{Y\});] \\
                create PUTable for t where \textit{Flag} = 1; \quad[\textit{Flag} = 0;] \\
                \If{\rm $i$ equal to \textit{MaxIS}$(s)$}{
                    \textit{UExtend}$(r, s)$ = 0; \\
                }
                add $t$ to \textit{rSet} and update the PUTable of $t$;\     
	    }
	}
	
	\For{$r$ $\in$ rSet}{
            calculate the \textit{conf}$(r)$ from its PUTable;\\
	    \If{\rm  \textit{tu}(r) $\ge$ \textit{minUtil} and \textit{conf}$(r)$ $\ge$ \textit{minconf}}{
		        save $r$;
		}
          \If{\rm  \textit{PUTable}.\textit{UBPart} $\ge$ max(\textit{minUtil}, \textit{riseU}) and \textit{conf}$(t)$ $\ge$ \textit{minconf}(\textit{PUTable}.\textit{UBPart} $\ge$ max(\textit{minUtil}, \textit{riseU}))}{
		    call \textbf{SingleExpansion}($r$, \textit{SDB}, \textit{IPEUM}, \textit{minUtil}, \textit{minconf}, \textit{minRatio});
		}
        }
\end{algorithm}

In the DoubleExpansion procedure, SRIU first creates an initial \textit{rSet}, which is a list containing all expanded rules that are generated from $r$, and it obtains the largest item in the subsequent $n$ of $r$. In the \textit{right-left} expansion, the goal is to obtain the largest item $m$ in the antecedent (line 1). Calculate the utility \textit{riseU} that meets the utility increasing ratio (line 2). For each sequence containing $r$, if the \textit{UExtend} of $r$ in $s$ is equal to 0, which means no item can be extended in this sequence, then stop the expansion of $r$ in $s$ (lines 4-6). If the item $i$ can be used in the left expansion in $r$, which is denoted as $i \in \textit{LE}$, similarly, in the \textit{right-left} expansion, the range of extended items is $i \in \textit{RE}$. According to \textit{IPEUP}, if \textit{IPEUM}$(i,n)$ is less than the largest of \textit{minUtil} and \textit{riseU}, SRIU stops the expansion $r$ with $i$. As for the \textit{right-left} expansion, the condition of \textit{IPEUP} is $\textit{IPEUM}(m,i)$ (lines 8-10). SRIU lets $t$ be the sequential rule extended by $r$ adds $i$, and updates the UTable for $t$ where \textit{Flag} equals 1, but \textit{Flag} is 0 in the \textit{right-left} (lines 11-12). Moreover, if $i$ is the largest item in $s$, then $r$ could not execute the left expansion, which means that $\textit{UL}(r,s)$ should be set to 0, and the items that could perform both left and right expansion are only used for right expansion. For the same reason, in the \textit{right-left} expansion, $\textit{UR}(r,s)$ = 0 and $\textit{ULR}(r,s)$ is changed (lines 13-15). And then, add $t$ to the rSet and update the UTable of $t$ (line 16). Finally, for each sequential rule $t$ in \textit{rSet}, if it is the high-utility sequential rule and satisfies the increasing utility ratio, that saves it (lines 21-23). According to \textit{EULEP} (in the \textit{right-left} is \textit{REEUP}), only the \textit{UBTotal} of $t$ satisfies the conditions, then $t$ could execute the DoubleExpansion procedure continually (lines 24-26). Similarly, only if \textit{UBPart} of $t$ equals or is greater than the largest of \textit{minUtil} and \textit{riseU} and the confidence of $t$ satisfies the conditions of \textit{CONFP} that $t$ could perform right expansion. Note that \textit{CONFP} is forbidden in \textit{right-left} expansion (lines 27-29).

The SingleExpansion procedure is similar to the DoubleExpansion procedure; there are only three differences. First, in order to properly use \textit{IPEUP}, the SingleExpansion procedure obtains the largest item $m$ in the antecedent of $r$; while during the \textit{right-left} expansion, it obtains the largest item $n$ in the consequent (line 1) and applies \textit{IPEUP} (lines 8-10). Secondly, because the SingleExpansion procedure only considers the left or right expansion, it uses PUTable to store information about $r$, which is more compact than UTable (line 12). In the end, since it is not allowed to perform the left expansion after the right expansion, when the expandable item $i$ is the maximum item in the sequence, the \textit{UExtend}$(r,s)$ can be set to 0. Similarly, this situation can also be applied in the \textit{right-left} expansion (lines 13-15).

\subsection{Complexity analysis}

In this section, we conduct a theoretical analysis of the time and space complexity of the proposed SRIU algorithm. We define $M$ as the number of sequences in the database \textit{SDB}, $L$ as the average sequence length (number of itemsets), and $|I|$ as the number of distinct items.

\subsubsection{Time complexity}

The time complexity of SRIU is analyzed phase by phase. Firstly, SRIU scans the database to compute the SEU for each item, yielding a complexity of $O(M \times L)$. It then constructs the IPEUM structure, which involves evaluating all $O(|I|^2)$ item pairs over the $M$ sequences, leading to a worst-case time of $O(|I|^2 \times M)$. Second, SRIU generates all promising $1 * 1$ sequential rules $r$ = $\{X\} \to \{Y\}$. The number of such candidate rules is $O(|I|^2)$. For each rule, building the TempTable requires processing all sequences containing it, resulting in a complexity of $O(|I|^2 \times M)$. Third, the core of SRIU lies in the rule expansion phase, which employs depth-first search with left-right or right-left expansion. The search space comprises all possible sequential rules formed by partitioning subsets of items into antecedent and consequent. In the worst case—where no pruning occurs—each of the $|I|$ items may reside in the antecedent, the consequent, or be absent from the rule. This results in a search space of size $O(3^{|I|})$, confirming the inherent combinatorial complexity of the problem. Thus, the overall worst-case time complexity of SRIU is $O(|I|^2 \times M + 3^{|I|})$. Although the proposed pruning strategies do not alter this asymptotic bound under adversarially crafted inputs, they drastically reduce the practical search space on real-world datasets, as demonstrated empirically in Section \ref{sec:experimental}.

\subsubsection{Space complexity}

The space complexity of SRIU is dominated by two components: the IPEUM structure and the utility tables for candidate rules. The IPEUM structure stores a utility value for each item pair, requiring $O(|I|^2)$ space. SRIU also maintains utility tables (TempTable, UTable, PUTable) for candidate rules. If $|R_c|$ is the number of candidate rules kept after pruning, and each rule's information consumes $O(M)$ space in the worst case (storing an entry per sequence), the space complexity for these structures is $O(|R_c| \times M)$. The value of $|R_c|$ is itself bounded by the exponential search space, but it is controlled in practice by the pruning strategies. The overall space complexity of SRIU is $O(|I|^2 + |R_c| \times M)$. The memory is primarily determined by the size of the IPEUM and the number of candidate rules retained during the mining process. The adoption of the Roaring Bitmap structure mitigates memory usage associated with the $M$ factor, particularly benefiting large and sparse datasets.
  
\section{Experimental Results}  \label{sec:experimental}

The experimental results of the SRIU algorithm can be divided into two parts. The first part is a performance evaluation of SRIU, and the second is an evaluation of its mining results. To comprehensively assess the impact of the different strategies proposed in this paper, we devised four different variants of SRIU, including SRIU$_{V1}$ with Roaring Bitmap, SRIU$_{V2}$ with Bitmap, SRIU$_{V3}$ absence of \textit{CONFP} strategy, and SRIU$_{V4}$ absence of \textit{IPEUP} strategy, respectively. SRIU$_{V1}$ is the complete algorithm. It is important to note that both SRIU$_{V3}$ and SRIU$_{V4}$ are realized by leveraging the underlying Roaring Bitmap. Besides, the increasing utility ratio $\textit{minRatio}$ is set to 0, which means the utility between sequential rules $\alpha$ and $\beta$ generated by $\alpha$ satisfies $u(\alpha) \leq u(\beta)$. It ensures that all mined rules satisfy a non-decreasing utility property, meaning the expansion of a rule never leads to a utility loss. This is already a valuable guarantee for applications like e-commerce recommendation or financial surveillance, where one wishes to avoid promoting downgrading paths or missing escalating risks. Most of our performance evaluations (Sections \ref{sec:performanceSRIU} and \ref{sec:evaluationSRIU}) use this setting to establish a baseline and facilitate a fair comparison of the algorithm's efficiency against its own variants.

\subsection{Datasets}

All algorithms are implemented in Java, and the source code of the algorithms is publicly available on GitHub: https://github.com/DSI-Lab1/SRIU. Experiments are conducted on a PC equipped with a 64-bit Windows 10 operating system, an 11th-generation 64-bit Core i7 processor with 2.5 GHz, and 32 GB of RAM.

To assess the efficacy and scalability of the SRIU algorithm, we subjected it to testing across two distinct categories of datasets: four real-world datasets (Bible, Kosarak10K, Leviathan, and Sign) and two synthetic datasets (Syn10K and Syn20K). These datasets are available on SPMF\footnote{SPMF: \url{http://www.philippe-fournier-viger.com/spmf/}}. Four real-world datasets contain single-item-based sequences, while the two synthetic datasets contain multiple-item-based sequences. Table \ref{tab:experiments} outlines the specifics of the datasets. The notations in Table \ref{tab:experiments} are represented respectively: 1) $M$ represents the number of sequences in \textit{SDB}, 2) m represents the number of itemsets in \textit{SDB}, 3) $|I|$ is the number of the distinct items in \textit{SDB}, 4) $L$ means the average size of each sequence in \textit{SDB}, 5) $n$ means the average size of each itemset in \textit{SDB}. Typically, higher values of $L$ and $n$ indicate a greater dataset density. To clearly describe the density of each dataset, we define the occurrence rate of each item in each dataset as \textit{density}(\textit{dataset}) = \textit{Avg}(\textit{occur(item)}) = $(\sum_{item \in \textit{SDB}}{ \frac{\textit{seq}(item)}{M}}) / |I|$. The \textit{density} of each testing dataset is shown in the column \textit{density} in Table \ref{tab:experiments}. Note that the parameter settings are based on the user's prior knowledge and the dataset's characteristics. 

\subsection{Performance of SRIU}
\label{sec:performanceSRIU}
\begin{table}[!ht]
    \fontsize{5.5pt}{9pt}\selectfont
    \caption{Information on experimental datasets}
    \label{tab:experiments}
    \centering
    \setlength{\tabcolsep}{3mm}{
	\begin{tabular}{ccccccc}
		  \hline
            \textbf{Dataset} & $M$ & $m$ & $|I|$ & $L$ & $n$ & \textbf{density} \\
		  \hline
		  Bible & 36,369 & 13,905 & 13905 & 17.85 & 1.00 & 0.12\%\\ 
		  Kosarak10k & 10,000 & 10,094 & 10,094 & 8.41 & 1.00 & 0.008\%\\
		  Leviathan & 5,834 & 9026 & 9025 & 26.34 & 1.00 & 0.29\%\\
		  Sign & 730 & 267 & 267 & 52.00 & 1.00 & 19.47\%\\
            Syn10K & 10,000 & 36,289 & 7,312 & 27.00 & 4.34 & 0.37\%\\
            Syn20K & 20,000 & 64,722 & 7,442 & 26.87 & 4.32 & 0.36\% \\
		  \hline
	\end{tabular}
    }
\end{table}

\subsubsection{Runtime analysis}

\begin{figure}[htbp]
    \centering
    \includegraphics[scale=0.46, trim=20 0 0 20]{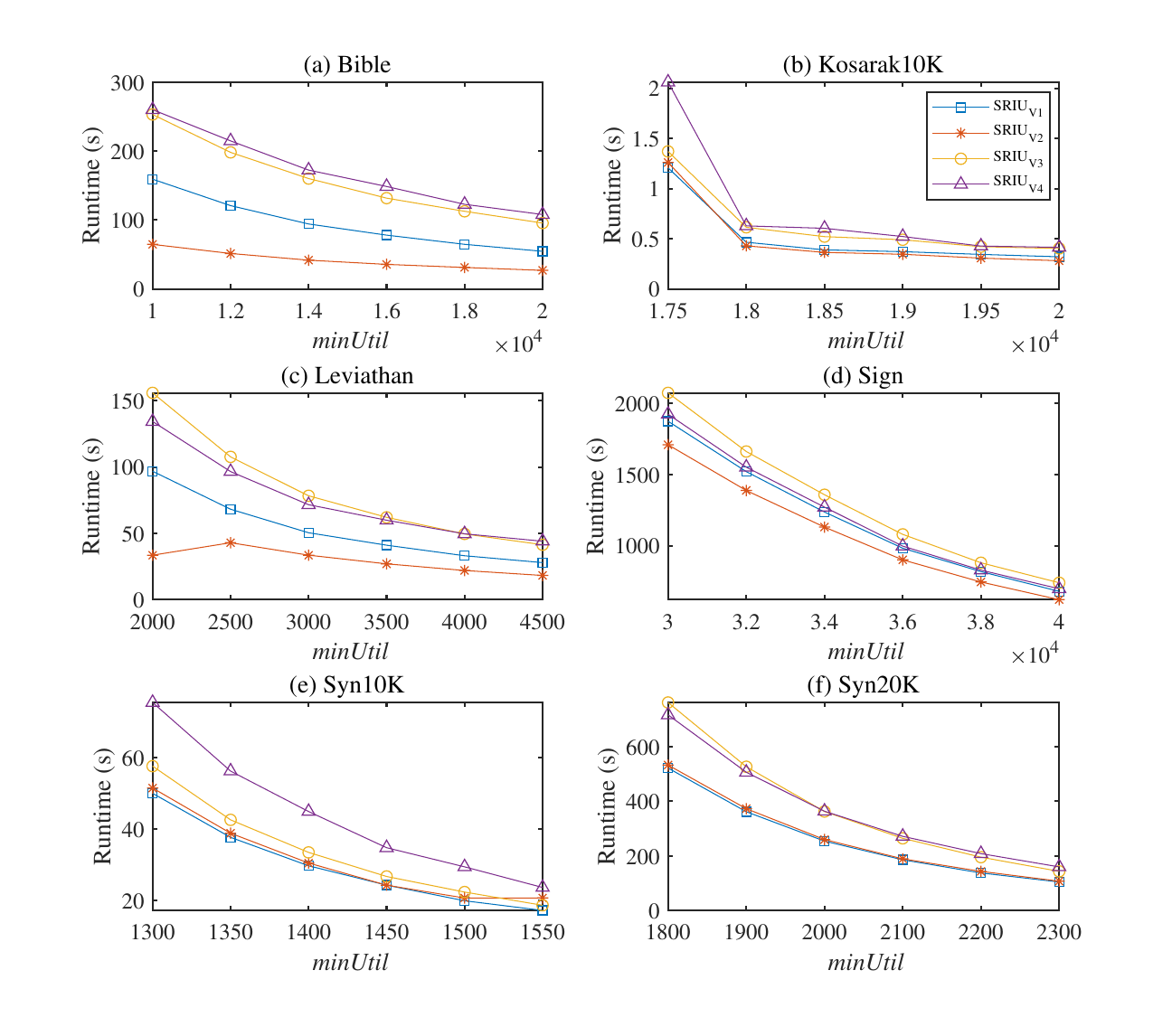}
    \caption{The runtime usage on different thresholds with increasing utility \textit{minRatio} = 0.}
    \label{fig:compare_runtime}
\end{figure}

Fig. \ref{fig:compare_runtime} shows the running time of each variant of SRIU. From a data structure perspective, it is clear that SRIU$_{V2}$ outperforms SRIU$_{V1}$ on the Bible and Leviathan datasets. Conversely, similar performance is observed in the remaining datasets. Roaring Bitmap incorporates three inherent container types: the ArrayContainer, the BitmapContainer, and the RunContainer, each serving specific scenarios within the framework. The default container in Roaring Bitmap is ArrayContainer, and it will be changed to BitmapContainer when the number of elements in Roaring Bitmap is greater than 4096. According to the frequency of items (the \textit{density (dataset)} can reveal that) within the dataset, it is discernible that, during the experimental phase, the container utilized in Roaring Bitmaps is the ArrayContainer. This container employs a binary search mechanism for executing union operations, resulting in a time complexity of $O(log N)$, where $N$ can be estimated as $N$ = \textit{density(dataset)} $\times$ $M$. For example, for the Bible dataset, $N$ = 0.12\% $\times$ 36369 = 43.6428. Compared with the O(1) time complexity of Bitmap, SRIU with Roaring Bitmap will need more time to update the \textit{sid} of sequential rules. Although the $N$ of Syn10K and Syn20K are 37 and 72, respectively, the running times of SRIU$_{V1}$ and SRIU$_{V2}$ are close. This outcome can be attributed to the principal temporal expenditure of both SRIU$_{V1}$ and SRIU$_{V2}$ during the sequence rule generation phase. The rationale behind this phenomenon lies in the significantly expanded search space inherent to synthetic datasets compared to their real-world counterparts.

As for SRIU$_{V3}$ and SRIU$_{V4}$, it is clear that the CONFP and IPEUP strategies are efficient. \textit{IPEUP} can reduce the redundant computation of both the slightly and moderately dense and sparse datasets, except for the Sign dataset. Sign is extremely dense since it only has 730 sequences, but the average size of each sequence is 52, and the average probability of occurrence of each item in Sign is 19.47\%, which means the estimated utility of item pairs is notably elevated. Consequently, the effectiveness of the \textit{IPEUM} strategy may not exhibit pronounced discernibility. As for the \textit{CONFP} strategy, since it is applied in the \textit{left-right} expansion process, the number of $1 * 1$ sequential rules to perform \textit{left-right} expansion is essential. Table \ref{tab:leftOrRightFirst} gives an example of the number of $1 * 1$ sequential rules to execute the expansion strategy of each dataset in a specific \textit{minUtil}. The common situation is that the more $1 * 1$ sequential rules ready to perform \textit{left-right} expansion, the more impact \textit{CONFP} strategy will have. In summary, the pruning strategies \textit{IPEUP} and \textit{CONFP} can effectively reduce SRIU's search space and improve its running speed. Furthermore, when applied to single-item-based sequences in these experiments, the Roaring Bitmap encounters suboptimal runtime performance due to its use of the ArrayContainer. Nevertheless, with the expansion of the search space manifested by an increase in the quantity and magnitude of sequences and itemsets, the Roaring Bitmap demonstrates a performance improvement.

\begin{table}[!ht]
    \fontsize{7pt}{9pt}\selectfont
    \caption{An example of the expansion strategy in each dataset in specific \textit{minUtil}.}
    \label{tab:leftOrRightFirst}
    \centering
    \setlength{\tabcolsep}{3mm}{
	\begin{tabular}{cccc}
		  \hline
            \textbf{Dataset} & \textit{minUtil} & $\textit{left-right}$ expansion & $\textit{right-left}$ expansion \\
		  \hline
		  Bible & 20000 & 7,640 & 843 \\ 
		  Kosarak10k & 17500 & 214 & 27 \\
		  Leviathan & 2000 & 8134 & 7618\\
		  Sign & 40000 & 2121 & 1667 \\
            Syn10K & 1300 & 17,115 & 16,611 \\
            Syn20K & 2300 & 19,764 & 17,133 \\
		  \hline
	\end{tabular}
    }
\end{table}

\subsubsection{Memory usage analysis}
\begin{figure}[htbp]
    \centering
    \includegraphics[scale=0.46, trim=30 30 0 30]{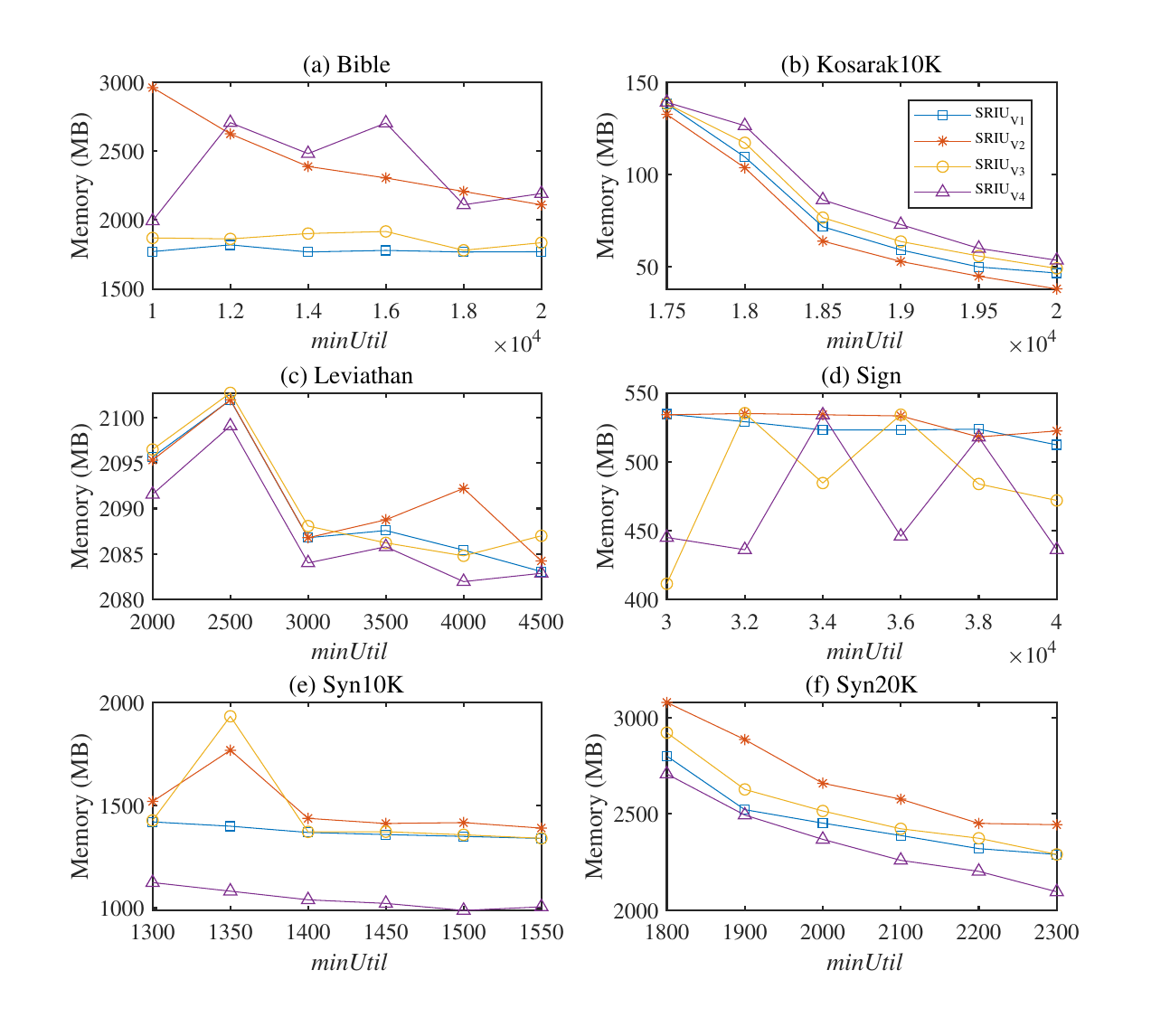}
    \caption{The consumption of memory on different thresholds with increasing utility \textit{minRatio} = 0.}
    \label{fig:compare_memory}
\end{figure}

As depicted in Fig. \ref{fig:compare_memory}, it is evident that SRIU$_{V1}$ exhibits superior memory usage compared to SRIU$_{V2}$. This outcome is attributed to the inherent nature that can effectively cope with the sparse storage of the Roaring Bitmap, which, unlike the Bitmap, obviates the necessity to allocate space equivalent to the dataset size for \textit{sid} storage for every item or sequential rule. Furthermore, a notable observation emerges: SRIU$_{V2}$ consistently exhibits the highest memory consumption on nearly all datasets, especially on datasets with a large number of sequences. These cases reflect Bitmap's limitations; even though the sequential rule only appears in 10 sequences, Bitmap still needs to allocate a bit of space equal to the size of the dataset. Therefore, if a dataset has more sequences and candidates, then the use of Bitmap needs to consume more memory. On the Kosarak10k dataset, the performance of SRIU$_{V2}$ emerges as optimal. This can be attributed to the efficacy of the \textit{SEUP} strategy, which effectively sieves out unpromising candidates. As a result, the memory overhead is notably minimized. From the $e$ and $f$ in Fig. \ref{fig:compare_memory}, we can find that SRIU$_{V4}$ consumes less memory than SRIU$_{V1}$. The rationale behind this phenomenon lies in the nature of the synthetic datasets Syn10K and Syn20K, which are characterized by their composition of multiple-item-based sequences and their large size. Consequently, the \textit{IPEUP} strategy necessitates the construction and storage of estimated utilities for a plethora of item pairs. This operation, in turn, contributes to substantial memory consumption. From Fig. \ref{fig:compare_memory}(a), we can see that the \textit{IPEUP} strategy can effectively reduce the search space and accelerate SRIU. On the datasets Kosarak10k and Leviathan, a striking similarity exists in memory consumption between SRIU$_{V4}$ and SRIU$_{V1}$. Furthermore, an examination of $b$ and $c$ in Fig. \ref{fig:compare_runtime} reveals that the \textit{IPEUP} strategy effectively enhances SRIU's performance. Although the \textit{CONFP} strategy can speed up SRIU, there is no obvious effect on saving memory overhead. The reason may be that we only consider the highest memory consumption of SRIU. In summary, Roaring Bitmap can reduce memory consumption, and \textit{IPEUP} performs better in terms of sacrificing memory for time.

\subsubsection{Scalability analysis}
\begin{figure}[ht]
    \centering
    \includegraphics[scale=0.36, trim=30 0 0 0]{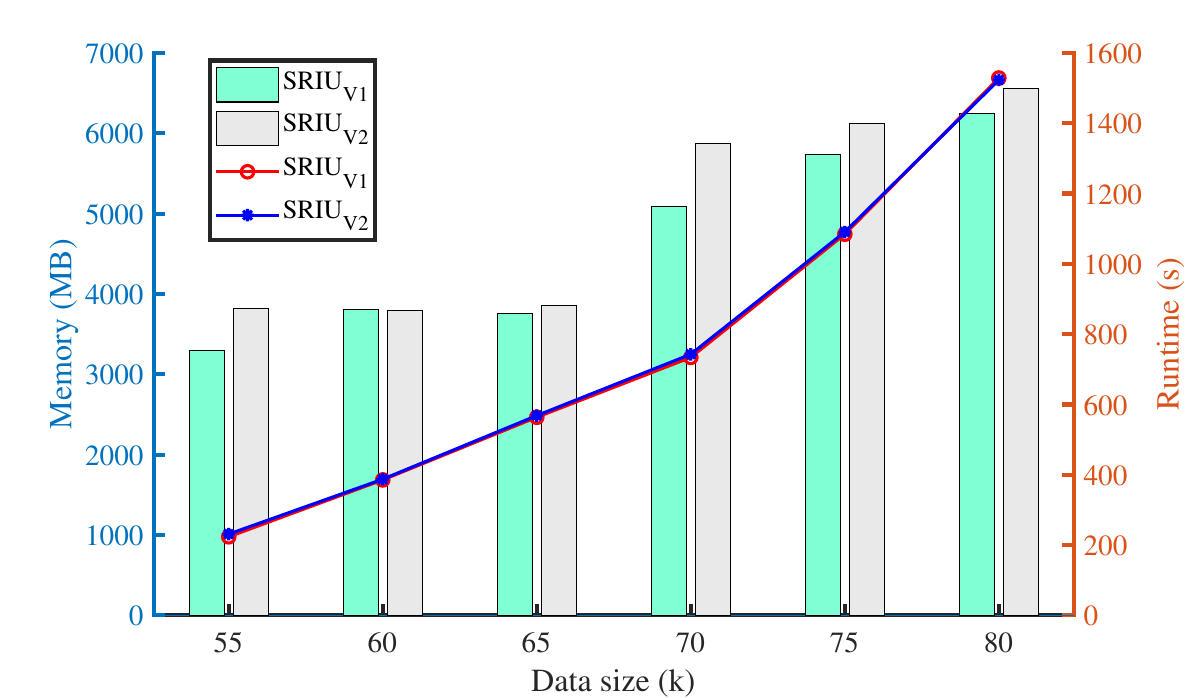}
    \caption{Scalability when \textit{minconf} = 0.6, \textit{minUtil} = 6000, and \textit{minRatio} = 0.}
    \label{fig:scalability}
\end{figure}

In this subsection, we selected synthetic datasets spanning sizes ranging from 55k to 80k to gauge the scalability of SRIU. We evaluated the scalability of SRIU in terms of running time and memory consumption. The default values of \textit{minconf}, \textit{minUtil}, and \textit{minRatio} are set to 0.6, 6000, and 0, respectively. The experimental results are shown in Fig. \ref{fig:scalability}. It is evident that as the dataset size increases, both the running time and memory consumption proportionally increase in a linear trend. Moreover, an assessment of the efficiency of SRIU$_{V1}$ and SRIU$_{V2}$ has been undertaken. As depicted in Fig. \ref{fig:scalability}, it is evident that SRIU$_{V1}$ yields enhanced memory conservation compared to SRIU$_{V2}$ while exhibiting comparable time performance.

\begin{figure}[htbp]
    \centering
    \includegraphics[scale=0.53, trim=30 40 0 20]{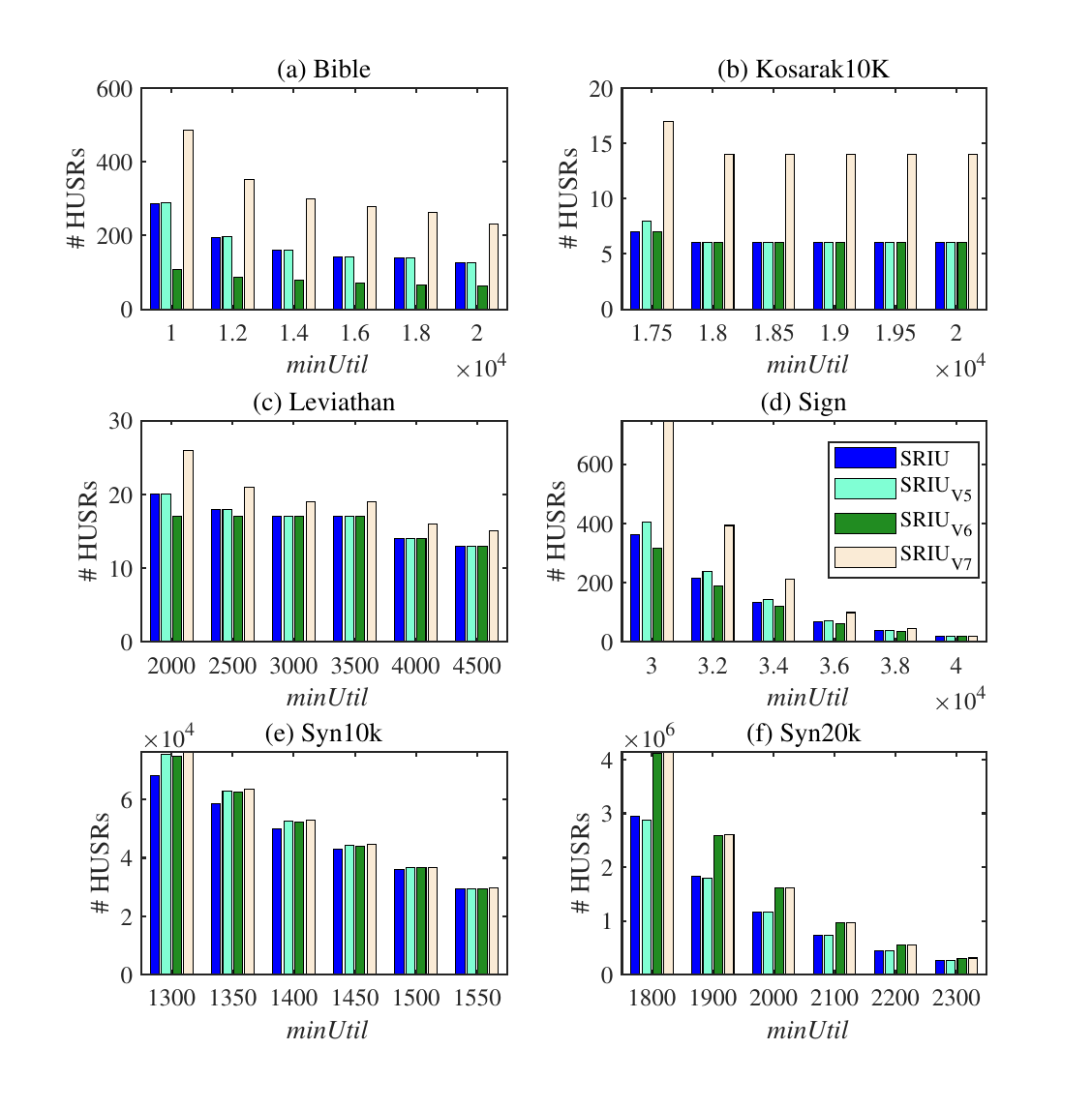}
    \caption{The number of HUSRs on different thresholds with increasing utility ratio = 0.}
    \label{fig:husrs}
\end{figure}

\subsection{Evaluation of mining results of SRIU}
\label{sec:evaluationSRIU}

We further evaluate SRIU's mining results using three key variants designed to isolate the impact of our core innovations:
\begin{itemize}
    \item \textbf{SRIU$_{V5}$}: use only the \textit{left-right} expansion strategy (fixed direction);
    \item \textbf{SRIU$_{V6}$}: use only the \textit{right-left} expansion strategy (fixed direction);
    \item \textbf{SRIU$_{V7}$}: disables the increasing utility ratio constraint, thereby solving the standard HUSRM problem akin to US-Rule \cite{huang2023us}.
\end{itemize}

The intuition is that SRIU$_{V7}$ can discover more HUSRs than other variants because it lacks the utility constraint that mandates adherence to an increasing ratio. From Fig. \ref{fig:husrs}, we can observe that the number of HUSRs mined by SRIU$_{V7}$ is significantly higher than other variants, especially on single-item-based datasets. However, on multiple-item-based datasets, Syn10K and Syn20K, the number of HUSRs of SRIU$_{V7}$ does not show significant differences compared to other variants. This phenomenon can be attributed to the fact that a substantial proportion of HUSRs can be ascertained through the incorporation of a novel item, specifically one situated within the same itemset as the last item within the antecedent or consequent. In this way, more HUSRs can reach the increasing utility ratio, especially if the ratio is set to 0. Furthermore, we can see from Fig. \ref{fig:husrs} that the number of HUSRs in SRIU is always lower than either SRIU$_{V5}$ or SRIU$_{V6}$. This is because SRIU comprehensively considers the two expansion methods and chooses different expansion methods according to the value of $e$-index. The $e$-index only considers the average utility, which can be used for left and right expansions of the sequential rule. In some cases, although $e$-index can help choose a better expansion method, the largest item will terminate the generation of sequential rules. Therefore, if we only take the \textit{left-right} expansion or \textit{right-left} expansion, it may generate more HUSRs. The following section will conduct an analytical examination of the validity of the mining outcomes produced by SRIU.

\begin{figure}[!htbp]
    \centering
    \includegraphics[scale=0.48, trim=40 0 0 0]{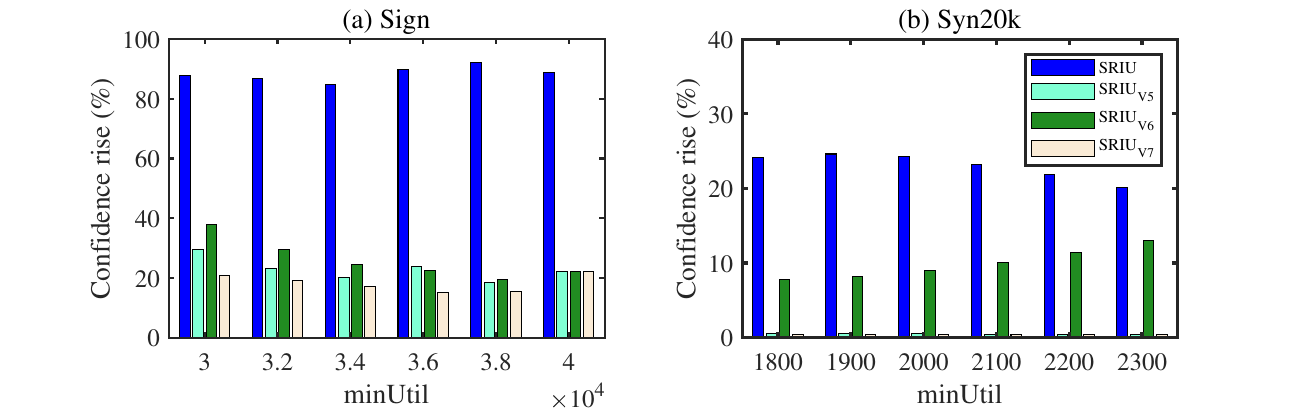}
    \caption{The confidence rise ratio of HUSRs on different thresholds with \textit{minRatio} = 0.}
    \label{fig:confRise}
\end{figure}

\begin{figure}[!htbp]
    \centering
    \includegraphics[scale=0.48, trim=40 0 0 0]{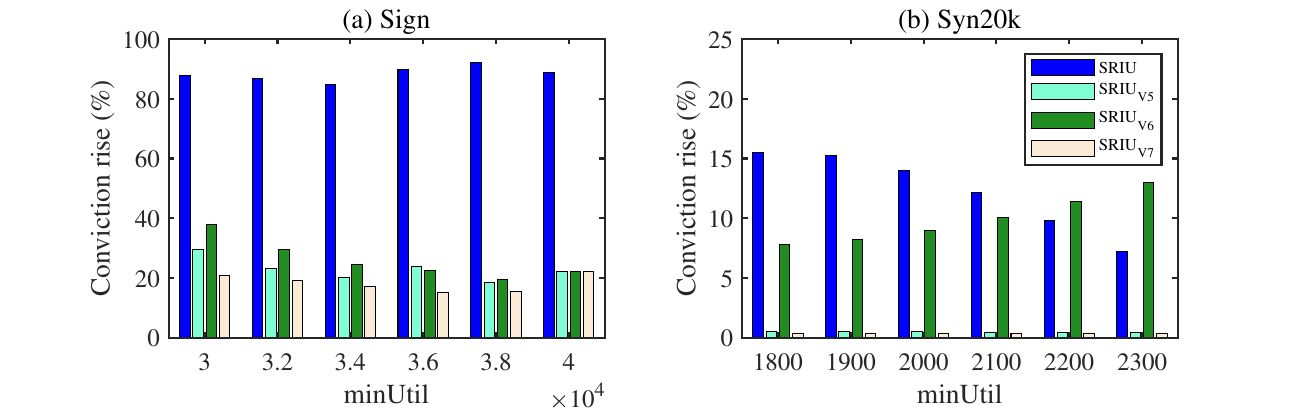}
    \caption{The conviction rise ratio of HUSRs on different thresholds with \textit{minRatio} = 0.}
    \label{fig:convRise}
\end{figure}

Note that there are two sequential rules, $\alpha$ and $\beta$, where $\beta$ is generated by $\alpha$ adding a new item. The method assesses the relationship between $\alpha$ and $\beta$ by evaluating indicators of confidence and conviction to determine whether any promotions have occurred. Consequently, the Sign and Syn20K datasets have been selected for assessing the correlation between $\alpha$ and $\beta$. In this analysis, the association between $\alpha$ and $\beta$ is gauged through the rate of enhancement of metrics encompassing confidence and conviction. First, we analyze it from a confidence standpoint. According to Fig. \ref{fig:confRise}, in the test datasets, the rate of confidence enhancement observed in SRIU exceeds that in its respective variants. Moreover, the confidence increase rate of SRIU, SRIU$_{V5}$, and SRIU$_{V6}$ is better than SRIU$_{V7}$, which means discovering HUSRs with an increasing utility ratio can help find more reliable HUSRs. Then, from Fig. \ref{fig:convRise}, the first observation is that SRIU, SRIU$_{V5}$, and SRIU$_{V6}$ all have higher conviction rates than SRIU$_{V7}$, which represents the mining HUSRs by adopting an increasing utility ratio as more authentic. In Sign, SRIU achieves the biggest increase in conviction; however, in Syn20K, in some specific utility thresholds, the increase in conviction of SRIU$_{V6}$ is the biggest. From Fig. \ref{fig:confRise}(b), we can observe that the rate of confidence increase decreases, and from Table \ref{tab:leftOrRightFirst}, SRIU adopts more \textit{left-right} expansion. According to the formulation of conviction, the increase of confidence shows a downward trend, which means a larger denominator, and \textit{left-right} expansion represents the molecular invariance for the most sequential rules during expansion. As a result, the SRIU's conviction shows a downward trend. Moreover, since we set the default threshold of confidence as 0.6 and based the conclusion in \cite{luna2018optimization} conviction changes the most in the confidence range of 0.6 to 1, in Fig. \ref{fig:convRise}(b), the conviction of SRIU has a noticeable change. In contrast, SRIU$_{V6}$ is on the rise.

Fig. \ref{fig:rC} illustrates the impact of the increasing utility ratio constraint on both rule quantity and quality over the Syn20K dataset. As \textit{minRatio} increases from 0.01 to 0.17, the number of mined HUSRs decreases by more than 56\%, reflecting the progressive filtering of extensions that fail to meet the non-decreasing utility requirement. Crucially, this reduction is accompanied by a consistent improvement in average confidence (Avg Confidence), from 0.725 to 0.741, indicating that the retained rules exhibit stronger predictive reliability. These results confirm that enforcing an increasing utility ratio does not merely restrict the search space but actively promotes the discovery of high-quality, actionable sequential patterns.

\begin{figure}[!htbp]
    \centering
    \includegraphics[scale=0.25, trim=50 30 0 20]{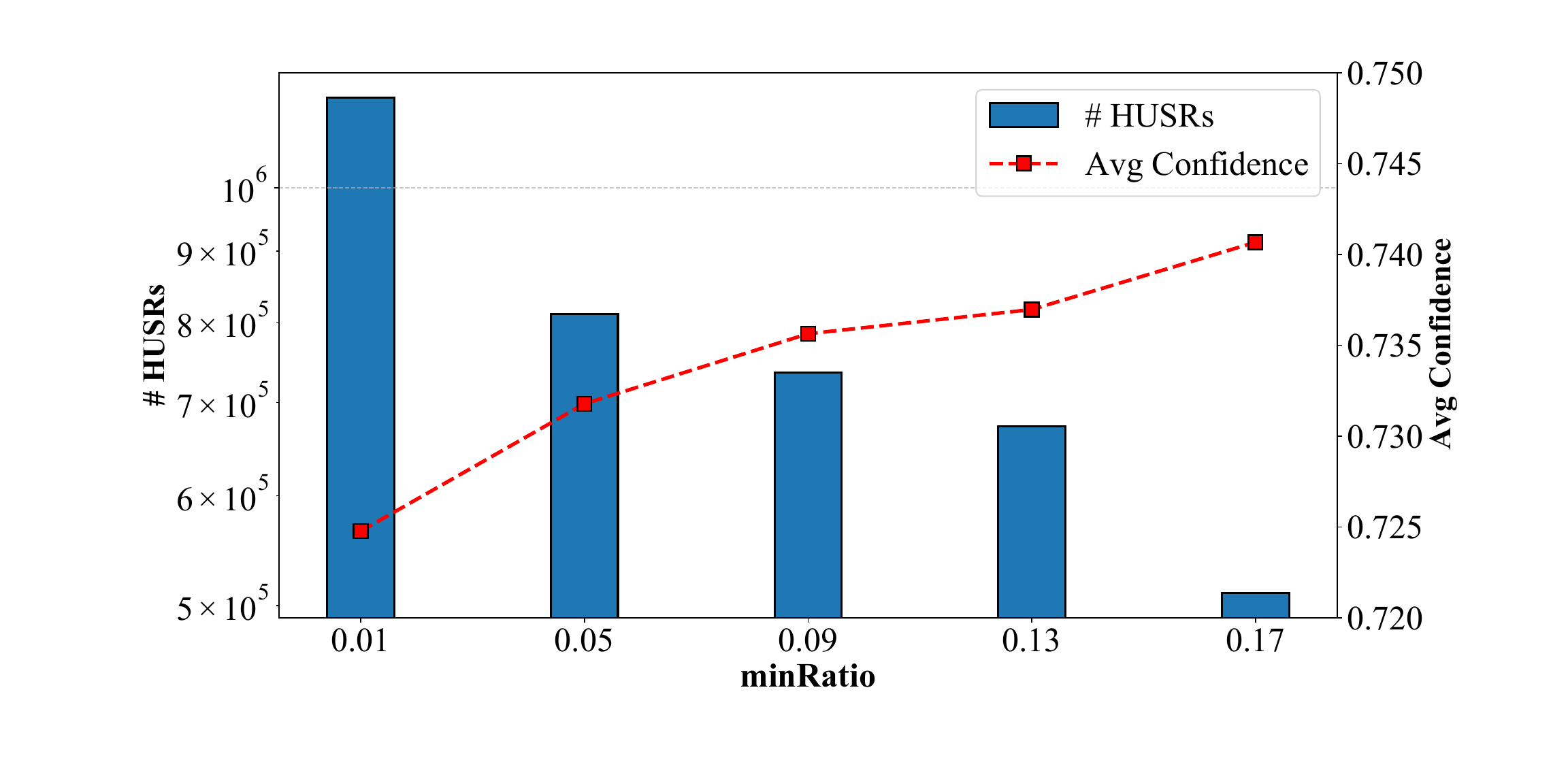}
    \caption{The total number of HUSRs and average confidence.}
    \label{fig:rC}
\end{figure}

\section{Conclusion} \label{sec:conclusion}
		
In this paper, we are the first to propose the SRIU algorithm to discover high-utility sequential rules with increasing utility ratios. A clear distinction is first shown between two distinct methods of expansion: \textit{left-right} expansion and \textit{right-left} expansion. Then, we proposed the $e$-index to determine the optimal expansion strategy. Moreover, we also designed two sets of upper bounds and the corresponding pruning strategies for the two types of expansion. Based on the data structure of IPEUM, SRIU can achieve better performance by pruning the unpromising sequential rules in advance. We also introduced the Roaring Bitmap to store sequence identity to reduce memory consumption, as well as the MaxIS strategy for issues with dense and long sequence datasets. Finally, we conducted a lot of experiments in different settings and evaluated the mining results. In the future, we want to further analyze the items that contribute to the increasing utility ratios since they play a positive role in the sequential rule. These influential items bear significant potential for integration into high-utility target pattern mining tasks, and their application extends to diverse fields, such as market basket analysis and emergency response systems.

\bibliographystyle{IEEEtran}
\bibliography{main}

\end{document}